\documentclass[sn-mathphys,Numbered]{sn-jnl}

\usepackage{graphicx}%
\usepackage{multirow}%
\usepackage{amsmath,amssymb,amsfonts}%
\usepackage{amsthm}%
\usepackage{mathrsfs}%
\usepackage[title]{appendix}%
\usepackage{xcolor}%
\usepackage{textcomp}%
\usepackage{manyfoot}%
\usepackage{booktabs}%
\usepackage{algorithm}%
\usepackage{algorithmicx}%
\usepackage{algpseudocode}%
\usepackage{listings}%
\usepackage{makecell}
\usepackage{bm}
\usepackage{array}

\theoremstyle{thmstyleone}%
\newtheorem{theorem}{Theorem}
\newtheorem{proposition}[theorem]{Proposition}%
\newtheorem{lemma}[theorem]{Lemma} 

\theoremstyle{thmstyletwo}%
\newtheorem{example}{Example}%
\newtheorem{remark}{Remark}%

\theoremstyle{thmstylethree}%
\newtheorem{definition}{Definition}%

\begin{document}

\title[Article Title]{New QEC codes and EAQEC codes from repeated-root cyclic codes
	of length $2^rp^s$ }


\author[1]{\fnm{Lanqiang} \sur{Li}}\email{lilanqiang716@126.com}

\author[1]{\fnm{Ziwen} \sur{Cao}}\email{caozwiii@163.com}

\author*[1]{\fnm{Tingting} \sur{Wu}}\email{wutingting58@163.com}

\author[2]{\fnm{Li} \sur{Liu}}\email{liuli-1128@163.com}

\affil[1]{\orgdiv{School of Information and Artificial Intelligence}, \orgname{Anhui Agricultural University}, \orgaddress{ \city{Hefei}, \postcode{230036}, \state{Anhui}, \country{China}}}
\affil[2]{\orgdiv{School of Mathematics}, \orgname{HeFei University of Technology}, \orgaddress{ \city{Hefei}, \postcode{230009}, \state{Anhui}, \country{China}}}

\abstract{Let $p$ be an odd prime and $r,s,m$ be positive integers. In this study, we initiate our exploration by delving into the intricate structure of all repeated-root cyclic codes and their duals with a length of $2^rp^s$ over the finite field $\mathbb{F}_{p^m}$. Through the utilization of CSS and Steane's constructions, a series of new quantum error-correcting (QEC) codes are constructed with parameters distinct from all previous constructions. Furthermore, we provide all maximum distance separable (MDS) cyclic codes of length $2^rp^s$, which are further utilized in the construction of QEC MDS codes. Finally, we introduce a significant number of novel entanglement-assisted quantum error-correcting (EAQEC) codes derived from these repeated-root cyclic codes. Notably, these newly constructed codes exhibit parameters distinct from those of previously known constructions. }

\keywords{Repeated-root cyclic codes, Maximum distance separable codes, Quantum error-correcting codes, Entanglement-assisted quantum error-correcting codes}

\pacs[MSC Classification]{94B15,94B05,11T71}

\maketitle

\section{Introduction}\label{sec1}
Cyclic code is a commonly used error correction coding technique, typically employed in data transmission and storage to protect data from errors that may occur during transmission or storage processes. The key features of cyclic codes are their ability to detect and correct errors, as well as having simple encoding and decoding algorithms. 

Repeated-root cyclic codes, a special type of cyclic codes, were first investigated by Castagnoli\cite{Castagnoli1991} and van Lint\cite{van1991} in the 1990s. Their research revealed that these codes possess a concatenated structure and are proved asymptotically bad. Nonetheless, it is recognized in the academic community that optimal repeated-root cyclic codes do exist, as demonstrated in various studies (refer to \cite{Dinh2008,Dinh2010,Kai2010}), which has encouraged many scholars to delve deeper into exploring this category of codes. Repeated-root cyclic codes are also used in communication and storage systems, especially in scenarios that require high data reliability. By utilizing the characteristics of duplicate roots, repeated-root cyclic codes can provide more powerful error detection and correction capabilities, thereby ensuring the integrity and reliability of data during transmission or storage. 

Quantum error-correcting (QEC) code is a coding technique used to protect quantum information from errors in qubits. With the development of quantum technology, QEC codes play a key role in the field of quantum computing and communication, and provide an important support for realizing reliable quantum information processing \cite{Shor1995,Steane1996,Knill1997,Rains1997,Rains1998,Steane1999}. By using QEC codes, researchers can improve the fault tolerance of quantum systems, thus promoting the application and development of quantum technology.

Indeed, since the inception of quantum coding theory, the construction of high-performance QEC codes has been a central goal pursued by numerous scholars in the field of quantum information science. At the end of the 20th century, Calderbank et al. \cite{Calderbank1998} established a connection between QEC codes and classical codes, proving that QEC codes can be constructed from classical linear codes with certain dual-containing properties. This systematic construction method is called as CSS construction. In the next five years, a series of QEC codes with better parameters were constructed based on CSS construction through classical linear codes, such as quantum Reed-Muller codes, quantum Reed-Solomon codes, quantum BCH codes, and quantum algebraic geometric codes. 

Maximum Distance Separable (MDS) codes are commonly used in coding theory and have good error correction ability. There are also corresponding QEC MDS codes in quantum computing and quantum communication, which are usually used to provide strong error correction capabilities and enhance security. QEC MDS codes have a wide range of applications in quantum error correction and quantum protocol security, which makes them become the most important research topic in recent years. Many scholars and researchers have dedicated their efforts to studying and advancing the construction of QEC MDS codes, leading to significant progress and numerous valuable results (refer to \cite{Grassl2004,Jin2010,Kai2013,Kai2014,Chen2015,Jin2017,Shi2017,Fang2020,Ball2021,Jin2022}
and the references therein for details).

In the construction process of QEC codes, it is critical to ensure that the selected classical linear code meets specific dual-containing conditions. However, it is known that most linear codes do not satisfy this condition. In order to break through this bottleneck, in 2006, Brun, Devetak and Hsieh \cite{Brun2006} proposed for the first time to construct QEC codes by sharing entangled states between the receiver and the sender. This method broke the limitation that the original method must require dual-containing conditions, so that any linear code can be used to construct QEC codes. These QEC codes constructed by this method are called entanglement-assisted quantum error-correcting (EAQEC) codes. Since then, this new type of QEC codes have attracted the attention of many scholars (refer to \cite{Wilde2008,Hsieh2010,Li2011,Lai2013,Guo2013,Fan2016,Chen2017,Chen2018,Lu2018,Guenda2018,Li2019,Fang2019,Wang2020,Guenda2020,Jin2021} and the references therein for details). 
 
Recently, repeated-root codes have been widely used to construct QEC codes, and many new QEC codes have been obtained. In \cite{Dinhps}, Dinh et al. first obtained all MDS constantcyclic codes of length $p^s$ through the study of the structure of such codes. Furthermore, they obtained all QEC MDS codes by using such MDS codes. In the same year, in \cite{Dinh2ps}, based on the structure of cyclic and negacyclic codes of length $2p^s$, the authors constructed some new QEC MDS codes and quantum synchronizable codes (QSCs). Similarly, in \cite{Dinh4ps,Dinh6ps,Dinh5ps}, the authors constructed a large number of new QEC MDS codes and QSCs by considering cyclic or negacyclic codes with lengths of $4p^s,5p^s$ and $6p^s$. In \cite{Dinh3ps}, Dinh et al. continued this research and constructed many QEC codes and QSCs by using cyclic codes with a length of $3p^s$ according to CSS and Steane's constructions. In addition, they also gave all QEC MDS codes from such codes based on CSS construction. In their conclusion, they propose that favorable outcomes may be achieved through the utilization of repeated-root constcyclic codes with more generally lengths. More recently, in \cite{Liu10ps}, Liu and Hu gave two methods for constructing high-performance QEC codes, and thus obtained many new QEC codes and EAQEC codes through cyclic codes of length $10p^s$.

It can be seen that repeated-root codes are important resources for constructing high-performance QEC codes. Motivated by the previous work, in this work, we first determine the generator polynomials of cyclic codes of length $2^rp^s$ and their duals. Second, we give all MDS cyclic codes with a length of $2^rp^s$. As an application, we construct a series of new QEC codes and EAQEC codes according to CSS and Steane's constructions from such cyclic codes and their duals.
In addition, all QEC MDS codes are obtained based on CSS construction from these repeated-root cyclic codes.

The remainder of this paper is structured in the following manner. In Section \ref{sec2}, some common symbols are introduced firstly, and then the definitions and related concepts of cyclic codes, QEC codes and EAQEC codes are introduced. In Section \ref{sec3}, the structure of all repeated-root cyclic codes and their duals with a length of $2^rp^s$ over $F_q$ are determined. In Section \ref{sec4}, based on CSS and Steane's constructions, a series of new QEC codes are constructed by cyclic codes of length $2^rp^s$. In Section \ref{sec5}, all MDS cyclic codes of length $2^rp^s$ are provided, which are further used to construct QEC MDS codes. In Section \ref{sec6}, according to the hulls of these repeated-root cyclic codes, a series of new EAQEC codes are constructed. Finally, Section \ref{sec7} summarizes the research results of this paper.

\section{Preliminaries}\label{sec2}

Let's start with some common notations as follows:
\begin{enumerate}
	\item[$\bullet$] $q=p^m$ is a prime power, where $p$ is an odd prime and $m$ is a positive integer.   
	\item[$\bullet$] $\mathbb{F}_{q}$ denotes the finite field with $q$ elements, $\mathbb{F}_{q}^{*}$ denotes the multiplicative group of nonzero elements of $\mathbb{F}_{q}$.
	\item[$\bullet$] $\mathbb{F}_{q}^n$ is the $\mathbb{F}_{q}$ vector space of $n$-tuples.
	\item[$\bullet$] Denote the multiplicative order of $q$ modulo $e$ by $\textrm{ord}_e(q)$, where $e$ is some integer.
	\item[$\bullet$] For any $t\in T=\{0,1,2,\cdots,p^s-1\}$, denote $P_t=\prod\limits^{s-1}_{i=0}(t_m+1)$, where $t_i$ is the coefficient the $p$-adic representation of $t$.
	\item[$\bullet$] Denote $T_n=\{0,2^{r-1}\}\bigcup\{\bigcup_{i=2}^{r}(2^{r-i}D_i)\}$ when $q=2^ab+1$; and $T_n=\{0,2^{r-1},2^{r-2}\}\bigcup\{\bigcup_{i=3}^{r}(2^{r-i}E_i)\}$ or $\{0,2^{r-1},2^{r-2}\}\bigcup\{\bigcup_{i=3}^{r}(2^{r-i}O_i)\}$ when $q=2^ab-1$, where $n=2^r$, $D_i$, $E_i$ and $O_i$ are defined in Proposition \ref{prop2} and Proposition \ref{prop3}, respectively.
\end{enumerate}

Any nonempty subset $C$ of $\mathbb{F}_{q}^n$ is known as a code of length $n$ over $\mathbb{F}_{q}$. If $C$ is a vector space over $\mathbb{F}_{q}$ with dimension $k$, it will be referred to as an $[n,k,d]_q$ linear code, where $d$ is its Hamming distance. Furthermore, $C$ is defined as an $[n,k,d]_q$ cyclic code if the cyclic shift of each codeword of $C$ still belongs to $C$. If $\textrm{gcd}(n,p)=1$, $C$ is called as a single-root cyclic code; otherwise, it is said to be a repeated-root cyclic code. As we all known, every cyclic code $C$ can be regarded as an ideal $\langle g(x)\rangle$ of the quotient ring $\mathbb{F}_{q}[x]/(x^n-1)$, where $g(x)|(x^n-1)$ is a unique monic polynomial. That is $C= \langle g(x)\rangle$, where $g(x)$ is called the generator polynomial of $C$. 

For any two vectors $\overrightarrow{x},\overrightarrow{y}\in \mathbb{F}_{q}^n$, $\overrightarrow{x}\ast\overrightarrow{y}$ is recorded as the Euclidean inner product. The set $\{\overrightarrow{x}\in\mathbb{F}_{q}^n|\overrightarrow{x}\ast\overrightarrow{y}=0, \forall \overrightarrow{y}\in C\}$ is defined as the dual code of $C$, and is denoted as $C^{\bot}$. For any $[n,k,d]_q$ cyclic code $C$, its dual code is an $[n,n-k,d^{\bot}]_q$ cyclic code. When $C= \langle g(x)\rangle$, the polynomial $h(x)=(x^n-1)/g(x)$ is called as the  parity check  polynomial of $C$. Denote $h(x)^{*}=h(0)^{-1}x^{\textrm{deg}(h(x))}h(\frac{1}{x})$. It is well known that the cyclic code $C^{\bot}$ can be generated by $h(x)^{*}$, i.e., $C^{\bot}=\langle h(x)^{*}\rangle$. In addition, if $C$ and its dual code satisfy the relation $C^{\bot}\subseteq C$, we call it a dual-containing code. We denote $\textrm{Hull}(C)=C\cap C^{\bot}$, and is called as the hull of code $C$. Obviously, $C$ is a dual-containing code if $\textrm{Hull}(C)=C^{\bot}$.

Since the factorization of  $x^n-1$ is closely related to $q$-cyclotomic coset, we need to introduce some concepts and some important conclusions about $q$-cyclotomic coset. Denote $Z_n=\{0,1,2,\cdots,n-1\}$,  which is the ring of integers
modulo $n$. For any integer $s\in Z_n$, denote $n_s$ is the the least positive integer such that $sq^{n_s}\equiv s(\textrm{mod }n)$, then the set $C_{(s,n)}=\{s,sq,\cdots,sq^{n_s}\}$ is known as the $q$-cyclotomic coset modulo $n$ containing $s$. Denote the cardinality of the $q$-cyclotomic coset $C_{(s,n)}$ as $|C_{(s,n)}|$. Obviously, $|C_{(s,n)}|=n_s$. One can choose a random number in $C_{(s,n)}$ as the representative element of the set $C_{(s,n)}$. The random number is generally the smallest number in $C_{(s,n)}$. We denote the set of all representative elements of the $q$-cyclotomic cosets modulo $n$ by $T_n$, and is called as a complete set of representatives. Obviously, $Z_n=\bigcup_{i\in T_n}C_{(i,n)}$, and $C_{(i,n)}\bigcap C_{(j,n)}=\emptyset$ for $i,j\in T_n$. If $\beta$ is a primitive $n$-th root of unity, for any integer $s\in Z_n$, the polynomial $m_s(x)=\prod_{i\in C_{(s,n)}}(x-\beta^i)\in \mathbb{F}_{q}[x] $ is known as  the minimal polynomial of $\beta^s$, and  is irreducible over $\mathbb{F}_{q}$. Furthermore, the polynomial $x^n-1$ can be completely decomposed into the product of some minimal polynomials as $ x^n-1=\prod_{s\in T_n}m_s(x).$

With the development of classical error-correcting codes, QEC codes emerged and were further developed. There are similar basic concepts and basic ideas between QEC codes and classical error-correcting codes. In the following, we mainly introduce the basic concepts of QEC codes and EAQEC codes over finite fields (refer to \cite{Shor1995,Steane1996,Knill1997,Rains1997,Rains1998,Brun2006,Wilde2008,Hsieh2010}). 

Let $\mathcal{C}^q$ be a $q$-dimensional complex vector space, where $q$ is a power of a prime $p$. Let $\{|x\rangle:x\in \mathbb{F}_q\}$ be an orthogonal basis of $\mathcal{C}^q$. Then, for any $a,b\in \mathbb{F}_q$, the unitary operators $X_a$ and $Z_b$ on the complex vector space $\mathcal{C}^q$ can be defined as:  $$X_a|x\rangle=|x+a\rangle,Z_b|x\rangle=\varsigma^{\textrm{Tr}(bx)}|x\rangle,$$ where $\textrm{Tr}$ is the trace function from the extension field $\mathbb{F}_q$ to the prime field $\mathbb{F}_p$, and $\varsigma=\textrm{exp}(2\pi i/p)$ is a primitive $p$-th  root of unity. Let $V_n=(\mathcal{C}^q)^{\otimes n}=\mathcal{C}^q\otimes\mathcal{C}^q\otimes\cdots\otimes\mathcal{C}^q$, i.e., the
$n$-fold tensor product on $\mathcal{C}^q$. Then, $V_n$ has a basis given by:
$\{|\overrightarrow{x}\rangle=|x_1\rangle\otimes|x_2\rangle\otimes\cdots\otimes|x_n\rangle,\overrightarrow{x}=(x_1,x_2,\cdots,x_n)\in \mathbb{F}_q^n\}$.
Extending the unitary operators on the complex vector space to $V_n$, we have: $$\{X_a|\overrightarrow{x}\rangle=X_{a_1}|x_1\rangle\otimes X_{a_2}|x_2\rangle\otimes\cdots\otimes X_{a_n}|x_n\rangle=|\overrightarrow{x}+a\rangle,$$ $$\{Z_b|\overrightarrow{x}\rangle=Z_{b_1}|x_1\rangle\otimes Z_{b_2}|x_2\rangle\otimes\cdots\otimes Z_{b_n}|x_n\rangle=|\varsigma^{\textrm{Tr}(\overrightarrow{x}\cdot \overrightarrow{b})}|\overrightarrow{x}\rangle,$$ where $\overrightarrow{x}\cdot \overrightarrow{b}=\sum_{i=1}^{n}x_ib_i\in F_q$ is the inner product on $\mathbb{F}_q^n$. Additionally, the set of all error operators on $V_n$ is defined as: 
$$E_n=\{\varsigma^i X_{\overrightarrow{a}}Z_{\overrightarrow{b}}:0\leq i\leq p-1,\overrightarrow{a}=(a_1,\cdots,a_n)\in \mathbb{F}_q^n,\overrightarrow{b}=(b_1,\cdots,b_n)\in F_q^n\},$$ where $X_{\overrightarrow{a}}=X_{a_1}\otimes X_{a_2}\otimes\cdots\otimes X_{a_n}, Z_{\overrightarrow{b}}=Z_{b_1}\otimes Z_{b_2}\otimes\cdots\otimes Z_{b_n}$.\\

In the following, we describe the specific definition of QEC codes, which has been given in \cite{Rains1998}.

\begin{definition}\label{def1}
 The complex subspace $\mathcal{Q}$ of dimension $K\geq1$ of the Hilbert space $V_n$ is called an $[[n,k]]_q$ QEC code of length $n$ and dimension  $K$, where $k=\textrm{log}_qK$. 
\end{definition}

Similarly to classical error-correcting code, each QEC code also has a minimum distance $d$, which reflects the ability of the code to detect and correct errors. If a QEC code has minimum distance $d$, then it is denoted as $[[n,k,d]]_q$. Various bounds exist for the parameters of QEC codes, including the quantum Singleton bound\cite{Knill1997,Rains1997}, which is outlined as follows. 
\begin{theorem}\label{theo1}
For any $[[n,k,d]]_q$ QEC code, its parameters must satisfy: $k\leq n-2d+2$. The code is referred to as a QEC MDS code if the equality holds.
\end{theorem}
EAQEC code is a new coding scheme to improve the performance of QEC code through the interaction of qubit and entangled state. By adding appropriate quantum entanglement resources, EAQEC codes can provide better error correction capability than traditional error correction codes, while reducing the quantum resources required for error correction. This use of entangled states makes EAQEC codes more efficient and allows them to play an important role in areas such as quantum communication and quantum computing.
Compared to QEC codes, EAQEC codes have an additional important parameter: the pairs of maximally entangled states $c$. A QEC code is said to be an $[[n,k,d,c]]_q$ EAQEC code over $\mathbb{F}_q$ if it encodes $k$ logical qubits into $n$ physical qubits with the help of $c$ pairs of maximally entangled state. Moreover, $\frac{k}{n}$ and $\frac{k-c}{n}$ are called rate and net rate respectively, which are important parameters to measure the performance of a EAQEC code. Obviously, EAQEC codes is an important generalization of QEC codes.

\section{ Cyclic codes of length $2^rp^s$ and their duals  } \label{sec3}

As $q $ is a odd prime power, then $q\equiv 1(\textrm{mod }4)$ or $q\equiv -1(\textrm{mod }4)$. So, we can let $q=2^ab+1$ or $q=2^ab-1$, where $b$ is odd and $a\geq 2$. Let $n=2^r$. For the set $T_n$, we have the following results, which can be obtained in \cite{Sharma}. \\

\begin{proposition}\label{prop2}
	Let $a\geq 2$ be an integer and $b$ be an odd number. Denote $n=2^r$, we have the following results.\\
	(1) If $q=2^ab+1$, then $T_n=\{0,2^{r-1}\}\bigcup\{\bigcup_{i=2}^{r}(2^{r-i}D_i)\}$, where $$D_i=\begin{cases} 
		\{\pm1,\pm3,\pm3^2,\dots,\pm3^{2^{a-2}-1}\},&\:if\:a+1\leq i\leq r;\\
		\{\pm1,\pm3,\pm3^2,\dots,\pm3^{2^{i-2}-1}\},&\:if\:2\leq i\leq a.
	\end{cases}$$\\
	(2) If $q=2^ab-1$, then $T_n=\{0,2^{r-1},2^{r-2}\}\bigcup\{\bigcup_{i=3}^{r}(2^{r-i}E_i)\}$, where $$E_i=\begin{cases} 
		\{1,3,3^2,\dots,3^{2^{a-1}-1}\},&\:if\:a+1\leq i\leq r,a\geq3;\\
		\{1,3,3^2,\dots,3^{2^{i-2}-1}\},&\:if\:3\leq i\leq a,a\geq3;\\
		\{\pm1\},&\:if\:3\leq i\leq r, a=2.
	\end{cases}$$
\end{proposition}

When $q=2^ab-1$, one can transform the set $T_n$ into the following form, which will greatly promote the study of dual codes.\\

\begin{proposition}\label{prop3}
	Let $a\geq 2$ be an integer and $b$ be an odd number. Denote $n=2^r$. If $q=2^ab-1$, then $T_n=\{0,2^{r-1},2^{r-2}\}\bigcup\{\bigcup_{i=3}^{r}(2^{r-i}O_i)\}$, where $$O_i=\begin{cases} 
		\{\pm1,\pm3,\pm3^2,\dots,\pm3^{2^{a-2}-1}\},&\:if\:a+1\leq i\leq r,a\geq3;\\
		\{1,3,3^2,\dots,3^{2^{i-2}-1}\},&\:if\:3\leq i\leq a,a\geq3;\\
		\{\pm1\},&\:if\:3\leq i\leq r, a=2.
	\end{cases}$$
\end{proposition}
\begin{proof} 
Comparing with Conclusion 2 in Proposition \ref{prop2}, it is not difficult to find that we only need to prove that $C_{(2^{r-i}U,n)}\neq C_{(2^{r-i}V,n)}$ for any $U,V\in Q_i$, when $a\geq3$ and $a+1\leq i\leq r$. Recall that $E_i=\{1,3,3^2,\dots,3^{2^{a-2}-1},3^{2^{a-2}},\cdots,3^{2^{a-1}-1}\}$ with $a\geq3$ and $a+1\leq i\leq r$. We know that all the $C_{(2^{r-i}3^{v},n)}$ with $0\leq v\leq 2^{a-2}-1$ are distinct. This implies that  all the $C_{(-2^{r-i}3^{v},n)}$ with $0\leq v\leq 2^{a-2}-1$ are also distinct. Next, we prove that there are no $u$ and $v$, $0\leq u,v\leq2^{a-2}-1$ such that $C_{(2^{r-i}3^{v},n)}=C_{(-2^{r-i}3^{u},n)}$. To that end, we assume $C_{(2^{r-i}3^{v},n)}=C_{(-2^{r-i}3^{u},n)}$ for some $0\leq u,v\leq2^{a-2}-1$. Note that $\textrm{ord}_{2^{r}}(q)=2^{r-a}$. Then, there exist some $0\leq j\leq 2^{r-a}-1$ such that 
$2^{r-i}3^{v}q^{j}\equiv-2^{r-i}3^{u}\pmod{2^r}$, where $u$ and $v$ are some integers satisfying $0\leq u,v\leq2^{a-2}-1$. That is $3^{v}q^{j}\equiv-3^{u}\pmod{2^i}$.
Since $a+1\leq i$ and $q^{2j}\equiv1\pmod{2^{a+1}}$, we can deduce that $3^{2(u-v)}\equiv1\pmod{2^{a+1}}$. 

Based on the fact that $\textrm{ord}_{2^{a+1}}(3)=2^{a-1}$, we can obtain $2^{a-1}|2(u-v)$. It then follows from that $u=v$. Hence, we have 
\begin{eqnarray} \label{equ 1}
q^{j}\equiv-1\pmod{2^i}.
\end{eqnarray} 

 As $q=2^ab-1$, it is obvious that $q^2\equiv1\pmod{2^{a+1}}$. Then, for any even $j$, it is easy to see that $q^j\equiv1\pmod{2^{a+1}}$. For any odd $j$, denote $j=2j^{'}+1$, we then have $q^j=q^{2j^{'}+1}\equiv q \equiv2^a-1\pmod{2^{a+1}}$.
When $i=a+1$, it follows from (\ref{equ 1}) that $q^{j}\equiv-1\pmod{2^{a+1}}$, which clearly contradicts the discussed arguments. When $i\geq a+2$, it follows from (\ref{equ 1}) that $q^{2j}\equiv1\pmod{2^{i}}$. Since $\textrm{ord}_{2^{i}}(q)=2^{i-a}$ and $4|2^{i-a}$, we can easily get $4$ divides $2j$, which means $j$ must be an even number. However, when $j$ is even, $q^{j}\equiv-1\pmod{2^i}$ clearly not true. From the above discussions, we can see that our assumption is wrong. Therefore, for any $U,V\in Q_i$, we obtain $C_{(2^{r-i}U,n)}\neq C_{(2^{r-i}V,n)}$.
\end{proof}	

According to the above propositions, one can factorize the polynomial $x^{2^r}-1$ in $\mathbb{F}_{q}[x]$. And then, the structure of all repeated-root cyclic codes and their duals with a length of $2^rp^s$ over $\mathbb{F}_q$ can be obtained.\\

\begin{theorem}\label{theo4}
	Let the notations be as described above. Suppose that $C$ is a cyclic code of length $2^rp^s$ over $\mathbb{F}_{q}$. We have the following results.\\
	(1) If $q=2^ab+1$, then $$C= \textbf{\Big\langle}\textbf{\big(} m_0(x)\textbf{\big)}^{j_0}\textbf{\big(}m_{2^{r-1}}(x)\textbf{\big)}^{j_{2^{r-1}}}\prod_{i=2}^r\prod_{s\in D_i}\textbf{\big(}m_{s2^{r-i}}(x)\textbf{\big)}^{j_{s2^{r-i}}}\textbf{\Big\rangle},$$ $$\textrm{ and }C^{\bot}=  \textbf{\Big\langle}\textbf{\big(} m_0(x)\textbf{\big)}^{p^s-j_0}\textbf{\big(}m_{2^{r-1}}(x)\textbf{\big)}^{p^s-j_{2^{r-1}}}\prod_{i=2}^r\prod_{s\in D_i}\textbf{\big(}m_{s2^{r-i}}(x)\textbf{\big)}^{p^s-j_{-s2^{r-i}}}\textbf{\Big\rangle},$$ where $0\leq j_h\leq p^s$ for each $h$.\\
	(2) If $q=2^ab-1$, then $$C= \textbf{\Big\langle}\textbf{\big(} m_0(x)\textbf{\big)}^{j_0}\textbf{\big(}m_{2^{r-1}}(x)\textbf{\big)}^{j_{2^{r-1}}}\textbf{\big(}m_{2^{r-2}}(x)\textbf{\big)}^{j_{2^{r-2}}}\prod_{i=3}^r\prod_{s\in O_i}\textbf{\big(}m_{s2^{r-i}}(x)\textbf{\big)}^{j_{s2^{r-i}}}\textbf{\Big\rangle},$$  where $0\leq j_h\leq p^s$ for each $h$. Moreover, when $a=2$, we have 
	$$C^{\bot}= \textbf{\Big\langle}\textbf{\big(} m_0(x)\textbf{\big)}^{p^s-j_0}\textbf{\big(}m_{2^{r-1}}(x)\textbf{\big)}^{p^s-j_{2^{r-1}}}\textbf{\big(}m_{2^{r-2}}(x)\textbf{\big)}^{p^s-j_{2^{r-2}}}$$ $$\prod_{i=3}^r\prod_{s\in O_i}\textbf{\big(}m_{s2^{r-i}}(x)\textbf{\big)}^{p^s-j_{-s2^{r-i}}}\textbf{\Big\rangle}.$$ 
	When $a\geq3$, we have 
	$$C^{\bot}= \textbf{\Big\langle}\textbf{\big(} m_0(x)\textbf{\big)}^{p^s-j_0}\textbf{\big(}m_{2^{r-1}}(x)\textbf{\big)}^{p^s-j_{2^{r-1}}}\textbf{\big(}m_{2^{r-2}}(x)\textbf{\big)}^{p^s-j_{2^{r-2}}}$$ 
	$$\prod_{i=3}^a\prod_{s\in O_i}\textbf{\big(}m_{s2^{r-i}}(x)\textbf{\big)}^{p^s-j_{s2^{r-i}}}\prod_{i=a+1}^r\prod_{s\in O_i}\textbf{\big(}m_{s2^{r-i}}(x)\textbf{\big)}^{p^s-j_{-s2^{r-i}}}\textbf{\Big\rangle}.$$ 
\end{theorem}
\begin{proof} (1) When $r=1$, it is clear that $C= \textbf{\Big\langle}\textbf{\big(} m_0(x)\textbf{\big)}^{j_0}\textbf{\big(}m_{1}(x)\textbf{\big)}^{j_{1}}\textbf{\Big\rangle},$ where $m_0(x)=(x-1)$ and $m_1(x)=(x+1)$. When $r\geq 2$, by the Conclusion (1) in Proposition \ref{prop2}, we have $T_n=\{0,2^{r-1}\}\bigcup\{\bigcup_{i=2}^{r}(2^{r-i}D_i)\}$ if $q=2^ab+1$. Then, we obtain $$x^{2^r}-1=m_0(x)m_{2^{r-1}}(x)\prod_{i=2}^r\prod_{s\in D_i}m_{s2^{r-i}}(x).$$ 
According to the definition of the set $D_i$, we know that for any $s\in D_i$, $-s$ still belongs to $D_i$. This implies that $s$ runs through $D_i$, and $-s$ runs through $D_i$. Combined with the definition of reciprocal polynomials, we have $m^{*}_0(x)=m_0(x)$, $m^{*}_{2^{r-1}}(x)=m_{2^{r-1}}(x)$ and $m^{*}_{s2^{r-i}}(x)=m_{-s2^{r-i}}(x)$ for any $s\in D_i$. Moreover, note that $ \prod_{s\in D_i} \textbf{\big(}m_{-s2^{r-i}}(x)\textbf{\big)}^{p^s-j_{s2^{r-i}}}=\prod_{s\in D_i} \textbf{\big(}m_{s2^{r-i}}(x)\textbf{\big)}^{p^s-j_{-s2^{r-i}}}$. So we can determine the dual codes $C^{\bot}$.

(2) When $a=2$, from Proposition \ref{prop3}, we get $O_3,O_4,\dots,O_r=\{1,-1\}$. So, $T_n=\{0,2^{r-1},2^{r-2},\pm2^{r-r},\pm2^{r-(r-1)},\dots,\pm2^{r-3}\},$ where each $r-i$ is greater than  or equal to $0$. First, $C_{2^{r-2}}=\{2^{r-2},-2^{r-2}\}$ is easily verifiable. So, $m^{*}_{2^{r-2}}(x)=m_{2^{r-2}}(x)$. Second, it is clear that for any $s$ belonging to $O_i$,$-s$ still belongs to $O_i$, when $3\leq i\leq r$. This implies that $m^{*}_{s2^{r-i}}(x)=m_{-s2^{r-i}}(x)$ for any $s\in O_i$. Recall that $m^{*}_0(x)=m_0(x)$ and $m^{*}_{2^{r-1}}(x)=m_{2^{r-1}}(x)$. So we can determine the dual codes $C^{\bot}$ when $a=2$. When $a\geq 3$, it follows from Proposition \ref{prop3} that $O_i=\{1,3,3^2,\dots,3^{2^{i-2}-1}\}$ for any $3\leq i\leq a$. It is clear that $q\equiv-1\pmod{2^i}$, since $q=2^ab-1$ and $3\leq i\leq a$. Further, it is easy to verify that $2^{r-i}3^tq\equiv-2^{r-i}3^t\pmod{2^r}$ for any $3\leq i\leq a$. This indicates that $C_{-2^{r-i}s,n}=C_{2^{r-i}s,n}$, i.e., $m^{*}_{s2^{r-i}}(x)=m_{s2^{r-i}}(x)$, for any $s\in O_i$ with $3\leq i\leq a$. Moreover, by using the same method as in (1), we can obtain $m^{*}_{s2^{r-i}}(x)=m_{-s2^{r-i}}(x)$ for any $s\in O_i$ with $a+1\leq i\leq r$. To summarize the above discussion, we can come to the dual codes $C^{\bot}$ when $a\geq 3$.
\end{proof} 	


\section{QEC codes of length $2^rp^s$  } \label{sec4}
In this section, some QEC codes of length $2^rp^s$ will be constructed by using repeated-root cyclic codes over $\mathbb{F}_q$. To achieve this, we first introduce an important class of construction methods for constructing quantum codes, which is called CSS construction \cite{Calderbank1998}.

\begin{theorem}\label{theo5}
	Let $C$ be an $[n,k,d]_q$ linear code, and let $C^{'}$ be an $[n,k^{'},d^{'}]_q$ linear code. Denote the Hamming distance of the duals of $C^{'}$ by $d^{'\bot}$. Suppose that $C^{'}\subseteq C$, then a QEC code with parameters $[[n,k-k^{'},\textrm{min}\{d,d^{'\bot}\}]]_q$ is obtained. In particular, put $C^{'}=C^{\bot}$, then an $[[n,2k-n,d]]_q$ QEC code is obtained.
\end{theorem}

The CSS construction shows that a QEC code with parameters $[[n,2k-n,d,d]]_q$ can be obtained if there exists an $[n,k,d]_q$ linear code such that $C^{\bot}\subseteq C$. Moreover, in order to get the number of all dual-containing codes, we need the following lemma, which has been given in \cite{Dinh4ps}.

\begin{lemma}\label{lemm6}
	Let $n_1,n_2$ be two non-negative integers and $n$ be a positive integer. Then, there exist $\frac{(n+2)(n+1)}{2}$ pairs $\{n_1,n_2\}$ such that $n_1+n_2\leq n$.
\end{lemma}

Next, we begin to investigate the condition under which a cyclic code $C$ of length $2^rp^s$ over $\mathbb{F}_q$ is a dual-containing code.

\begin{proposition}\label{prop7}
	Let $q$ be a prime power and has the form $q=2^ab+1$. Suppose that $C$ is a cyclic code of length $2^rp^s$ over $\mathbb{F}_{q}$ and can be represented in the form $C= \textbf{\Big\langle}\textbf{\big(} m_0(x)\textbf{\big)}^{j_0}\textbf{\big(}m_{2^{r-1}}(x)\textbf{\big)}^{j_{2^{r-1}}}\prod\limits_{i=2}^r\prod\limits_{s\in D_i}\textbf{\big(}m_{s2^{r-i}}(x)\textbf{\big)}^{j_{s2^{r-i}}}\textbf{\Big\rangle}$, where $0\leq j_h\leq p^s$ for each $h$. Then, $C^{\bot}\subseteq C$ if and only if (briefly, iff) $0\leq j_0,j_{2^{r-1}}<\frac{p^s}{2}$  and $0\leq j_{s2^{r-i}}+j_{-s2^{r-i}}\leq p^s$ for any $s\in D_i$ with $2\leq i\leq r$. Moreover, the number of such codes satisfying $C^{\bot}\subseteq C$ equals $(p^s+2)^{2^{r-1}-1}(\frac{p^{s}+1}{2})^{2^{r-1}+1}$ if $r\leq a$, or  $(p^s+2)^{2^{a-1}+(r-a)2^{a-2}-1}(\frac{p^{s}+1}{2})^{2^{a-1}+(r-a)2^{a-2}+1}$ if $r> a$.   
\end{proposition}

\begin{proof} 
According to Theorem \ref{theo4}, we have $$C^{\bot}=  \textbf{\Big\langle}\textbf{\big(} m_0(x)\textbf{\big)}^{p^s-j_0}\textbf{\big(}m_{2^{r-1}}(x)\textbf{\big)}^{p^s-j_{2^{r-1}}}\prod_{i=2}^r\prod_{s\in D_i}\textbf{\big(}m_{s2^{r-i}}(x)\textbf{\big)}^{p^s-j_{-s2^{r-i}}}\textbf{\Big\rangle}.$$
Clearly,  $C^{\bot}\subseteq C$ iff $0\leq j_0\leq p^s-j_0,0\leq j_{2^{r-1}}\leq p^s-j_{2^{r-1}}$ and $0\leq j_{s2^{r-i}}\leq p^s-j_{-s2^{r-i}}$ i.e., $0\leq j_0,j_{2^{r-1}}<\frac{p^s}{2}$ and $0\leq j_{s2^{r-i}}+j_{-s2^{r-i}}\leq p^s$ for any $s\in D_i$ with $2\leq i\leq r$. According to Proposition \ref{prop2}, when $r\leq a$, we can see that there are $\sum_{i=2}^{r}2^{i-2}=2^{r-1}-1$ pairs $\{s2^{r-i},-s2^{r-i}\}$ in $T_n$. When $r>a$, there are $\sum_{i=2}^{a}2^{i-2}+\sum_{i=a+1}^{r}2^{a-2}=2^{a-1}+(r-a)2^{a-2}-1$ pairs $\{s2^{r-i},-s2^{r-i}\}$ in $T_n$. For each pair $\{s2^{r-i},-s2^{r-i}\}$, according to Lemma \ref{lemm6}, there are $(p^s+2)\frac{p^s+1}{2}$ values to choose $j_{s2^{r-i}},j_{-s2^{r-i}}$ such that $0\leq j_{s2^{r-i}}+j_{-s2^{r-i}}\leq p^s$. Moreover, it is clear that there are $\frac{p^s+1}{2}$ values to choose each of $j_0,j_{2^r-1}$. Therefore, we can obtain the desired result. 
\end{proof}

Based on Proposition \ref{prop3}, Theorem \ref{theo4}, and Lemma \ref{lemm6}, we can derive the following important result. Since the proof process is the same as Proposition \ref{prop7}, we omit it here.

\begin{proposition}\label{prop8}
	Let $q$ be a prime power and has the form $q=2^ab-1$. Suppose that $C$ is a cyclic code of length $2^rp^s$ over $\mathbb{F}_{q}$ and can be represented in the form $C= \textbf{\Big\langle}\textbf{\big(} m_0(x)\textbf{\big)}^{j_0}\textbf{\big(}m_{2^{r-1}}(x)\textbf{\big)}^{j_{2^{r-1}}}\textbf{\big(}m_{2^{r-2}}(x)\textbf{\big)}^{j_{2^{r-2}}}\prod\limits_{i=3}^r\prod\limits_{s\in O_i}\textbf{\big(}m_{s2^{r-i}}(x)\textbf{\big)}^{j_{s2^{r-i}}}\textbf{\Big\rangle}$,  where $0\leq j_h\leq p^s$ for each $h\in T_n$. Then, we have the following results.
		
 (1) When $a=2$, $C^{\bot}\subseteq C$ iff $0\leq j_0,j_{2^{r-1}},j_{2^{r-2}}<\frac{p^s}{2}$  and $0\leq j_{s2^{r-i}}+j_{-s2^{r-i}}\leq p^s$ for any $s\in O_i$ with $3\leq i\leq r$. Moreover, the number of such codes satisfying $C^{\bot}\subseteq C$ equals $(p^s+2)^{r-2}(\frac{p^{s}+1}{2})^{r+1}$.  

(2) When $a\geq3$, $C^{\bot}\subseteq C$ iff $0\leq j_0,j_{2^{r-1}},j_{2^{r-2}}<\frac{p^s}{2}$, $0\leq j_{s2^{r-i}}\leq \frac{p^s}{2}$ for any $s\in O_i$ with $3\leq i\leq a$ and $0\leq j_{s2^{r-i}}+j_{-s2^{r-i}}\leq p^s$ for any $s\in O_i$ with $a+1\leq i\leq r$. Moreover, the number of such codes satisfying $C^{\bot}\subseteq C$ equals $(p^s+2)^{(r-a)2^{a-2}}(\frac{p^{s}+1}{2})^{2^{a-1}+(r-a)2^{a-2}+1}$.  
\end{proposition}

According to CSS construction, we can obtain the following QEC codes of length $2^rp^s$ by using the above cyclic codes over $F_q$. 

\begin{theorem}\label{theo9}
	Let $q$ be a prime power and has the form $q=2^ab+1$. Suppose that $C= \textbf{\Big\langle}\textbf{\big(} m_0(x)\textbf{\big)}^{j_0}\textbf{\big(}m_{2^{r-1}}(x)\textbf{\big)}^{j_{2^{r-1}}}\prod\limits_{i=2}^r\prod\limits_{s\in D_i}\textbf{\big(}m_{s2^{r-i}}(x)\textbf{\big)}^{j_{s2^{r-i}}}\textbf{\Big\rangle}$ is a cyclic code of length $2^rp^s$ over $\mathbb{F}_{q}$ with
	 $0\leq j_0,j_{2^{r-1}}<\frac{p^s}{2}$  and $0\leq j_{s2^{r-i}}+j_{-s2^{r-i}}\leq p^s$ for any $s\in D_i$ with $2\leq i\leq r$. Then, we have the following results.
	
	(1) When $r\leq a$, a QEC code with parameters $[[2^rp^s,2^rp^s-2j_0-2j_{2^{r-1}}-2\sum_{i=2}^{r}\sum_{s\in D_i}j_{s2^{r-i}},d_H(C)]]_q$ is obtained. And, the number of QEC codes constructed based on CSS construction by using all cyclic codes of length $2^rp^s$ over $\mathbb{F}_q$ equals $(p^s+2)^{2^{r-1}-1}(\frac{p^{s}+1}{2})^{2^{r-1}+1}$.
	
	(2) When $r>a$, a QEC code with parameters $[[2^rp^s,2^rp^s-2j_0-2j_{2^{r-1}}-2\sum_{i=2}^{a}\sum_{s\in D_i}j_{s2^{r-i}}-2\sum_{i=a+1}^{r}\sum_{s\in D_i}(2^{i-a}j_{s2^{r-i}}),d_H(C)]]_q$ is obtained. And, the number of QEC codes constructed based on CSS construction by using all cyclic codes of length $2^rp^s$ over $\mathbb{F}_q$ equals $(p^s+2)^{2^{a-1}+(r-a)2^{a-2}-1}(\frac{p^{s}+1}{2})^{2^{a-1}+(r-a)2^{a-2}+1}$.     
\end{theorem}

\begin{proof} 
(1) Notice that $|C_{(2^{r-i}s,n)}|=\textrm{ord}_{\frac{2^{r-i}s}{\textrm{gcd}(2^{r-i}s,n)}}(q).$
For each $s\in D_i$ with $2\leq i\leq r$, we then have $|C_{(2^{r-i}s,n)}|=1$, which means the degree of $m_{s2^{r-i}}(x)$ is $1$. Hence, the cyclic code $C$ has parameters $[2^rp^s,2^rp^s-j_0-j_{2^{r-1}}-\sum_{i=2}^{r}\sum_{s\in D_i}j_{s2^{r-i}},d_H(C)]$. From Proposition \ref{prop7}, the cyclic code $C$ is a dual-containing code. Finally, based on CSS construction, we obtain a QEC code with parameters $[[2^rp^s,2^rp^s-2j_0-2j_{2^{r-1}}-2\sum_{i=2}^{r}\sum_{s\in D_i}j_{s2^{r-i}},d_H(C)]_q$ by using the cyclic code $C$. 

(2) Note that $\textrm{Ord}_{2^i}(q)=2^{i-a}$ when $i\geq a+1$. For each $s\in D_i$ with $a+1\leq i\leq r$, it is easy to verify that $|C_{(2^{r-i}s,n)}|=2^{i-a}$, which implies that the degree of $m_{s2^{r-i}}(x)$ is $2^{i-a}$. Recall that $|C_{(2^{r-i}s,n)}|=1$ for each $s\in D_i$ with $2\leq i\leq a$. Then the cyclic code $C$ has parameters $[2^rp^s,2^rp^s-j_0-j_{2^{r-1}}-\sum_{i=2}^{a}\sum_{s\in D_i}j_{s2^{r-i}}-\sum_{i=a+1}^{r}\sum_{s\in D_i}(2^{i-a}j_{s2^{r-i}}),d_H(C)]$.
Similarly, we get a QEC code with parameters $[[2^rp^s,2^rp^s-2j_0-2j_{2^{r-1}}-2\sum_{i=2}^{a}\sum_{s\in D_i}j_{s2^{r-i}}-2\sum_{i=a+1}^{r}\sum_{s\in D_i}(2^{i-a}j_{s2^{r-i}}),d_H(C)]]_q$ by using the cyclic code $C$. 
\end{proof}

Similarly, we can obtain the following theorem. The proof process is identical to that of Theorem \ref{theo9}, so it is omitted.

\begin{theorem}\label{theo10}
	Let $q$ be a prime power and has the form $q=2^ab-1$. Suppose that $C$ is a cyclic code of length $2^rp^s$ over $\mathbb{F}_{q}$ and can be represented in the form $C= \textbf{\Big\langle}\textbf{\big(} m_0(x)\textbf{\big)}^{j_0}\textbf{\big(}m_{2^{r-1}}(x)\textbf{\big)}^{j_{2^{r-1}}}\textbf{\big(}m_{2^{r-2}}(x)\textbf{\big)}^{j_{2^{r-2}}}\prod\limits_{i=3}^r\prod\limits_{s\in O_i}\textbf{\big(}m_{s2^{r-i}}(x)\textbf{\big)}^{j_{s2^{r-i}}}\textbf{\Big\rangle}$,  where $0\leq j_h\leq p^s$ for each $h\in T_n$. Then, we have the following results.

(1) When $a=2$, if $0\leq j_0,j_{2^{r-1}},j_{2^{r-2}}<\frac{p^s}{2}$ and $0\leq j_{s2^{r-i}}+j_{-s2^{r-i}}\leq p^s$ for any $s\in O_i$ with $3\leq i\leq r$, then a QEC code with parameters $[[2^rp^s,2^rp^s-2j_0-2j_{2^{r-1}}-4j_{2^{r-2}}-2\sum_{i=3}^{r}\sum_{s\in O_i}(2^{i-a}j_{s2^{r-i}}),d_H(C)]]_q$ is obtained. Moreover, the number of QEC codes constructed based on CSS construction by using all cyclic codes of length $2^rp^s$ over $\mathbb{F}_q$ equals $(p^s+2)^{r-2}(\frac{p^{s}+1}{2})^{r+1}$.  

(2) When $a\geq3$, if $0\leq j_0,j_{2^{r-1}},j_{2^{r-2}}<\frac{p^s}{2}$, $0\leq j_{s2^{r-i}}\leq \frac{p^s}{2}$ for any $s\in O_i$ with $3\leq i\leq a$ and $0\leq j_{s2^{r-i}}+j_{-s2^{r-i}}\leq p^s$ for any $s\in O_i$ with $a+1\leq i\leq r$, then a QEC code with parameters $[[2^rp^s,2^rp^s-2j_0-2j_{2^{r-1}}-4j_{2^{r-2}}-4\sum_{i=3}^{a}\sum_{s\in D_i}j_{s2^{r-i}}-2\sum_{i=a+1}^{r}\sum_{s\in D_i}(2^{i-a}j_{s2^{r-i}}),d_H(C)]]_q$ is obtained. Moreover,  the number of QEC codes constructed based on CSS construction by using all cyclic codes of length $2^rp^s$ over $\mathbb{F}_q$ equals $(p^s+2)^{(r-a)2^{a-2}}(\frac{p^{s}+1}{2})^{2^{a-1}+(r-a)2^{a-2}+1}$.     
\end{theorem}

\begin{example}\label{exam1}
 Put $a=4,b=1$, then $q=2^ab+1=17$ and $p=17$. Let $r=3,s=1$, then $n=2^rp^s=136$. By factorization, we have 
 $x^{136}-1=(x+1)^{17}(x-1)^{17}(x+2)^{17}(x-2)^{17}(x+4)^{17}(x-4)^{17}(x+8)^{17}(x-8)^{17}$.
 According to Theorem \ref{theo9}, the number of QEC codes constructed based on CSS construction by using all cyclic codes of length $136$ over $\mathbb{F}_{17}$ equals $115279213$. Some specific parameters are listed in Table \ref{tab1}.
 \begin{table}[!htbp]  \label{tab1}
 	\caption{\centering {Some QEC codes obtained from Conclusion (1) in Theorem \ref{theo9}} }
 	\begin{tabular}{|c|c|c|}
 		\hline
 		the generator polynomial $g(x)$ \textrm{ of } $C$ & the parameters of $C$ &  QEC code \\
 		\hline
 		$(x+1)$&  $[136,135,2]_{17}$  & $[[136,135,2]]_{17}$ \\
        \hline
 		$(x+1)^2(x+8)$&  $[136,133,3]_{17}$  & $[[136,130,3]]_{17}$ \\
 	    \hline
     	$(x+1)^{3}(x+2)(x+4)$&  $[136,131,4]_{17}$  & $[[136,126,4]]_{17}$ \\
        \hline
        	\makecell[c]{$(x+1)^{4}(x+2)$\\$(x+4)(x+8)^2$}&  $[136,128,5]_{17}$  & $[[136,120,5]]_{17}$ \\
        \hline
        	\makecell[c]{$(x+1)^{5}(x-1)$\\$(x+2)(x+4)(x+8)^{2}$}&  $[136,126,6]_{17}$  & $[[136,116,6]]_{17}$ \\
        \hline
       	\makecell[c]{ $(x+1)^{6}(x-1)(x+2)^2$\\$(x+4)(x+8)^3(x-8)$}&  $[136,122,7]_{17}$  & $[[136,108,7]]_{17}$ \\
        \hline
        	\makecell[c]{$(x+1)^{7}(x-1)^2(x+2)$\\$(x-2)(x+4)(x+8)^3(x-8)$}&  $[136,120,8]_{17}$  &  $[[136,104,8]]_{17}$ \\
        \hline
 	\end{tabular}
 \end{table}
\end{example}

\begin{example}\label{exam2}
 Put $a=2,b=3$, then $q=2^ab+1=13$ and $p=13$. Let $r=3,s=1$, then $n=2^rp^s=104$. By factorization, we have 
 $x^{104}-1=(x+1)^{13}(x-1)^{13}(x+5)^{13}(x-5)^{13}(x^2+5)^{13}(x^2-5)^{13}$.
According to Theorem \ref{theo9}, the number of QEC codes constructed based on CSS construction by using all cyclic codes of length $104$ over $\mathbb{F}_{13}$ equals $540225$. Some specific parameters are listed in Table \ref{tab2}. 

\begin{table} [h]  \label{tab2}
	\caption{\centering {Some QEC codes obtained from Conclusion (2) in Theorem \ref{theo9}} }
	\begin{tabular}{|c|c|c|}
		\hline
		the generator polynomial $g(x)$ \textrm{ of } $C$ & the parameters of $C$ &  QEC code \\
		\hline
		$(x+1)$&  $[104,103,2]_{13}$  & $[[104,102,2]]_{13}$ \\
		\hline
		$(x+1)^2(x^2+5)$&  $[104,100,3]_{13}$  & $[[104,96,3]]_{13}$ \\
		\hline
		$(x+1)^{3}(x^2+5)(x^2-5)$& $[104,97,4]_{13}$   & $[[104,90,4]]_{13}$  \\
		\hline
		$(x+1)^{4}(x+5)(x^2+5)^2(x^2-5)$&  $[104,93,5]_{13}$  & $[[104,82,5]]_{13}$ \\
		\hline
		$(x+1)^{5}(x+5)(x^2+5)^2(x^2-5)$&  $[104,92,6]_{13}$  & $[[104,80,6]]_{13}$ \\
		\hline
		$(x+1)^{6}(x-1)(x+5)^2(x^2+5)^3(x^2-5)$&  $[104,87,7]_{13}$  & $[[104,70,7]]_{13}$  \\
		\hline
		$(x+1)^7(x-1)(x+5)^2(x^2+5)^3(x^2-5))$&  $[104,86,8]_{13
		}$  & $[[104,68,8]]_{13}$ \\
		\hline
	\end{tabular}
\end{table}
\end{example}

\begin{example}\label{exam3}
 Put $a=2,b=3$, then $q=2^ab-1=11$ and $p=11$. Let $r=3,s=1$, then $n=2^rp^s=88$. By factorization, we have 
$x^{88}-1=(x+1)^{11}(x-1)^{11}(x^2+1)^{11}(x^2+3x+10)^{11}(x^2+8x+10)^{11}$.
According to Theorem \ref{theo9}, the number of QEC codes constructed based on CSS construction by using all cyclic codes of length $88$ over $\mathbb{F}_{11}$ equals $16848$. Some specific parameters are listed in Table \ref{tab3}. 
\begin{table}[!h]  \label{tab3}
	\caption{\centering {Some QEC codes obtained from Conclusion (1) in Theorem \ref{theo10}} }
	\begin{tabular}{|c|c|c|}
		\hline
		the generator polynomial $g(x)$ \textrm{ of } $C$ & the parameters of $C$ &  QEC code \\
		\hline
		$(x+1)$&  $[88,87,2]_{11}$  & $[[88,86,2]]_{11}$ \\
		\hline
	    $(x+1)^2(x^2+3x+10)$&  $[88,84,3]_{11}$  & $[[88,80,3]]_{11}$ \\
		\hline
		$(x+1)^3(x^2+3x+10)$& $[88,83,4]_{11}$   & $[[88,78,4]]_{11}$  \\
		\hline
		$(x+1)^4(x^2+1)(x^2+3x+10)^2$&  $[88,78,5]_{11}$  & $[[88,68,5]]_{11}$ \\
		\hline
		$(x+1)^5(x^2+1)(x^2+3x+10)^2$&  $[88,77,6]_{11}$  & $[[88,66,6]]_{11}$ \\
		\hline
		\makecell[c]{$(x+1)^3(x-1)(x^2+1)$\\$(x^2+3x+10)^6(x^2+8x+10)$}&  $[88,68,7]_{11}$  & $[[88,48,7]]_{11}$ \\
		\hline
		\makecell[c]{$(x+1)^3(x-1)(x^2+1)^2$\\$(x^2+3x+10)^7(x^2+8x+10)$}&  $[88,64,8]_{11
		}$  & $[[88,40,8]]_{11}$ \\
		\hline
		\makecell[c]{	$(x+1)^4(x-1)(x^2+1)^2$\\$(x^2+3x+10)^8(x^2+8x+10)$}&  $[88,61,9]_{11
		}$  & $[[88,34,9]]_{11}$ \\
		\hline
		\makecell[c]{	$(x+1)^3(x-1)^3(x^2+1)^4$\\$(x^2+3x+10)^{10}(x^2+8x+10)$}&  $[88,52,10]_{11
		}$  & $[[88,16,10]]_{11}$ \\
		\hline
		\makecell[c]{	$(x+1)^4(x-1)^4(x^2+1)^5$\\$(x^2+3x+10)^{10}(x^2+8x+10)$}&  $[88,48,11]_{11
		}$  & $[[88,8,11]]_{11}$ \\
		\hline
	\end{tabular}
\end{table}
\end{example}

\begin{example}\label{exam4}

	Put $a=3,b=1$, then $q=2^ab-1=7$ and $p=7$. Let $r=3,s=1$, then $n=2^rp^s=56$. By factorization, we have 
	$x^{56}-1=(x+1)^{7}(x-1)^{7}(x^2+1)^{7}(x^2+3x+1)^{7}(x^2+4x+1)^{7}$.
According to Theorem \ref{theo9}, the number of QEC codes constructed based on CSS construction by using all cyclic codes of length $56$ over $\mathbb{F}_{7}$ equals $1024$. Some specific parameters are listed in Table \ref{tab4}.
\begin{table}[!h]  \label{tab4}
	\caption{\centering {Some QEC codes obtained from Conclusion (2) in Theorem \ref{theo10}} }
	\begin{tabular}{|c|c|c|}
		\hline
		the generator polynomial $g(x)$ \textrm{ of } $C$ & the parameters of $C$ &  QEC code \\
		\hline
		$(x+1)$&  $[56,55,2]_{7}$  & $[[56,54,2]]_{7}$ \\
		\hline
		$(x+1)^2(x^2+3x+1)$&  $[56,52,3]_{7}$  & $[[56,48,3]]_{7}$ \\
		\hline
	 $(x+1)^3(x^2+3x+1)$& $[56,51,4]_{7}$   & $[[56,46,4]]_{7}$  \\
		\hline
	\end{tabular}
\end{table}
\end{example}

\begin{remark}\label{rem1}
The Hamming distances of all cyclic codes in Examples \ref{exam1},\ref{exam2},\ref{exam3}, \ref{exam4} can be determined using our latest research results in \cite{Pan2024}. Comparing these QEC codes in Examples \ref{exam1},\ref{exam2},\ref{exam3},\ref{exam4} and online table \cite{Grassl2007}, we find that our QEC codes are all new, meaning that they exhibit parameters distinct from those of previously known constructions. Specifically, our QEC codes in Example \ref{exam1} are better than all the QEC codes with the same length and Hamming distance given in \cite{Grassl2007}, that is, the dimension of our QEC codes is greater than the dimension of all the QEC codes with the same length and Hamming distance in \cite{Grassl2007}. In Example \ref{exam2}, our QEC codes $[[104,103,2]]$, $[[104,96,3]]$ and $[[104,80,6]]$ are better than all the QEC codes with the same length and Hamming distance given in \cite{Grassl2007}, and the other codes are consistent with  some known codes constructed in \cite{Grassl2007}. In Example \ref{exam3}, our QEC codes $[[88,86,2]]$, $[[88,78,4]]$, $[[88,68,5]]$ and $[[88,66,6]]$ are consistent with some known codes constructed in \cite{Grassl2007}. In Example \ref{exam4}, our QEC codes $[[56,54,2]]$ and $[[56,46,4]]$ are consistent with some known codes constructed in \cite{Grassl2007}. Note that these known QEC codes in \cite{Grassl2007} and our QEC codes in Examples \ref{exam1},\ref{exam2},\ref{exam3},\ref{exam4} are constructed over different fields. This demonstrates that these QEC codes derived in Examples \ref{exam1},\ref{exam2},\ref{exam3},\ref{exam4} are all novel.
\end{remark} 

Steane's construction method refers to a technique for QEC codes developed by Steane in 1999. This method is a famous way to construct QEC codes using classical linear codes, which is given in \cite{Steane1999}.

\begin{theorem}\label{theo11}
Let $C$ be an $[n,k,d]_q$ linear code, and let $C^{'}$ be an $[n,k^{'},d^{'}]_q$ linear code, where $k^{'}\geq k+1$. Suppose that $C$ is a dual-containing code and $C\subseteq C^{'}$, then a QEC code with parameters $[[n,k+k^{'}-n,\textrm{min}\{d,\lceil\frac{q+1}{q}\rceil d^{'}\}]]_q$ is obtained. 
\end{theorem}

According to Steane's construction, we can obtain the following QEC codes by using cyclic codes of length $2^rp^s$ over $\mathbb{F}_q$. 

\begin{theorem}\label{theo12}
	Let $q$ be a prime power and has the form $q=2^ab+1$. Suppose that $C= \textbf{\Big\langle}\textbf{\big(} m_0(x)\textbf{\big)}^{j_0}\textbf{\big(}m_{2^{r-1}}(x)\textbf{\big)}^{j_{2^{r-1}}}\prod\limits_{i=2}^r\prod\limits_{s\in D_i}\textbf{\big(}m_{s2^{r-i}}(x)\textbf{\big)}^{j_{s2^{r-i}}}\textbf{\Big\rangle}$ is a cyclic code of length $2^rp^s$ over $\mathbb{F}_{q}$ with
	$0\leq j_0,j_{2^{r-1}}<\frac{p^s}{2}$  and $0\leq j_{s2^{r-i}}+j_{-s2^{r-i}}\leq p^s$ for any $s\in D_i$ with $2\leq i\leq r$. Put $C^{'}=\textbf{\Big\langle}\textbf{\big(} m_0(x)\textbf{\big)}^{j^{'}_0}\textbf{\big(}m_{2^{r-1}}(x)\textbf{\big)}^{j^{'}_{2^{r-1}}}\prod\limits_{i=2}^r\prod\limits_{s\in D_i}\textbf{\big(}m_{s2^{r-i}}(x)\textbf{\big)}^{j^{'}_{s2^{r-i}}}\textbf{\Big\rangle}$, where $0\leq j^{'}_h\leq j_h$ for each $h$. The following results hold.
	
	(1) When $r\leq a$, then $k=2^rp^s-(j_0+j_{2^{r-1}}+\sum_{i=2}^{r}\sum_{s\in D_i}j_{s2^{r-i}})$ and $k^{'}=2^rp^s-(j^{'}_0+j^{'}_{2^{r-1}}+\sum_{i=2}^{r}\sum_{s\in D_i}j^{'}_{s2^{r-i}})$. Furthermore, if $k^{'}\geq k+1$, a QEC code with parameters $[[2^rp^s,k+k^{'}-2^rp^s,\textrm{min}\{d,\lceil\frac{q+1}{q}\rceil d^{'}\}]]_q$ is essential.
	
	(2) When $r>a$, then $k=2^rp^s-(j_0+j_{2^{r-1}}+\sum_{i=2}^{a}\sum_{s\in D_i}j_{s2^{r-i}}+\sum_{i=a+1}^{r}\sum_{s\in D_i}(2^{i-a}j_{s2^{r-i}}))$ and $k^{'}=2^rp^s-(j^{'}_0+j^{'}_{2^{r-1}}+\sum_{i=2}^{a}\sum_{s\in D_i}j^{'}_{s2^{r-i}}+\sum_{i=a+1}^{r}\sum_{s\in D_i}(2^{i-a}j^{'}_{s2^{r-i}}))$. Furthermore, if $k^{'}\geq k+1$, a QEC code with parameters $[[2^rp^s,k+k^{'}-2^rp^s,\textrm{min}\{d,\lceil\frac{q+1}{q}\rceil d^{'}\}]]_q$ is essential.   
\end{theorem}

Similarly, we can also obtain the following theorem.

\begin{theorem}\label{theo13}
	Let $q$ be a prime power and has the form $q=2^ab-1$. Suppose that $C$ is a cyclic code of length $2^rp^s$ over $\mathbb{F}_{q}$ and can be represented in the form $C= \textbf{\Big\langle}\textbf{\big(} m_0(x)\textbf{\big)}^{j_0}\textbf{\big(}m_{2^{r-1}}(x)\textbf{\big)}^{j_{2^{r-1}}}\textbf{\big(}m_{2^{r-2}}(x)\textbf{\big)}^{j_{2^{r-2}}}\prod\limits_{i=3}^r\prod\limits_{s\in O_i}\textbf{\big(}m_{s2^{r-i}}(x)\textbf{\big)}^{j_{s2^{r-i}}}\textbf{\Big\rangle}$,  where $0\leq j_h\leq p^s$ for each $h\in T_n$. Put $C^{'}= \textbf{\Big\langle}\textbf{\big(} m_0(x)\textbf{\big)}^{j^{'}_0}\textbf{\big(}m_{2^{r-1}}(x)\textbf{\big)}^{j^{'}_{2^{r-1}}}\textbf{\big(}m_{2^{r-2}}(x)\textbf{\big)}^{j^{'}_{2^{r-2}}}\prod\limits_{i=3}^r\prod\limits_{s\in O_i}\textbf{\big(}m_{s2^{r-i}}(x)\textbf{\big)}^{j^{'}_{s2^{r-i}}}\textbf{\Big\rangle}$, where $0\leq j^{'}_h\leq j_h$ for each $h$. Then, we have the following results.
	
	(1) When $a=2$, then $k=2^rp^s-(j_0+j_{2^{r-1}}+2j_{2^{r-2}}+\sum_{i=3}^{r}\sum_{s\in O_i}(2^{i-a}j_{s2^{r-i}}))\geq1$ and $k^{'}=2^rp^s-(j^{'}_0+j^{'}_{2^{r-1}}+2j^{'}_{2^{r-2}}+\sum_{i=3}^{r}\sum_{s\in O_i}(2^{i-a}j^{'}_{s2^{r-i}}))$, Furthermore, if $0\leq j_0,j_{2^{r-1}},j_{2^{r-2}}<\frac{p^s}{2}$ and $0\leq j_{s2^{r-i}}+j_{-s2^{r-i}}\leq p^s$ for any $s\in O_i$ with $3\leq i\leq r$, then a QEC code with parameters $[[2^rp^s,k+k^{'}-2^rp^s,\textrm{min}\{d,\lceil\frac{q+1}{q}\rceil d^{'}\}]]_q$ is essential for $k^{'}\geq k+1$. 
	
	(2) When $a\geq3$, then $2^rp^s-(j_0-j_{2^{r-1}}-2j_{2^{r-2}}-2\sum_{i=2}^{a}\sum_{s\in D_i}j_{s2^{r-i}}-\sum_{i=a+1}^{r}\sum_{s\in D_i}(2^{i-a}j_{s2^{r-i}}))\geq1$ and $2^rp^s-(j^{'}_0+j^{'}_{2^{r-1}}+2j^{'}_{2^{r-2}}+2\sum_{i=2}^{a}\sum_{s\in D_i}j^{'}_{s2^{r-i}}+\sum_{i=a+1}^{r}\sum_{s\in D_i}(2^{i-a}j^{'}_{s2^{r-i}}))$. Furthermore, if $0\leq j_0,j_{2^{r-1}},j_{2^{r-2}}<\frac{p^s}{2}$, $0\leq j_{s2^{r-i}}\leq \frac{p^s}{2}$ for any $s\in O_i$ with $3\leq i\leq a$ and $0\leq j_{s2^{r-i}}+j_{-s2^{r-i}}\leq p^s$ for any $s\in O_i$ with $a+1\leq i\leq r$, then a QEC code with parameters $[[2^rp^s,k+k^{'}-2^rp^s,\textrm{min}\{d,\lceil\frac{q+1}{q}\rceil d^{'}\}]]_q$ is essential for $k^{'}\geq k+1$. 
\end{theorem}

\begin{example}\label{exam5}
Let $a=4,b=1$, then $q=2^ab+1=17$ and $p=17$. Let $r=3,s=1$, then $n=2^rp^s=136$. By factorization, we have 
$x^{136}-1=(x+1)^{17}(x-1)^{17}(x+2)^{17}(x-2)^{17}(x+4)^{17}(x-4)^{17}(x+8)^{17}(x-8)^{17}$. Put $C=\langle g(x)m(x) \rangle$ and $C^{'}=\langle g(x)\rangle$, where $g(x)m(x)$ is a factor of $x^{136}-1$. Specific parameters of some codes are shown in table \ref{tab5}.
 \begin{table}[!htbp]  \label{tab5}
	\caption{\centering {Some QEC codes constructed from Conclusion (1) in Theorem \ref{theo12}} }
	\begin{tabular}{|c|c|c|c|c|}
		\hline
	 $g(x)$ & $m(x)$&  $C$ &  $C^{'}$ &   QEC code\\
		\hline
		$(x+1)$&  	$(x+1)(x+8)$& $[136,133,3]_{17}$  &$[136,135,2]_{17}$  & $[[136,132,3]]_{17}$ \\
		\hline
		$(x+1)^2(x+8)$&  	$(x+1)(x+4)$& $[136,131,4]_{17}$  &$[136,133,3]_{17}$   & $[[136,130,4]]_{17}$ \\
		\hline
	$(x+1)^{3}(x+2)(x+4)$&  	$(x+1)(x+8)^2$& $[136,128,5]_{17}$  &$[136,131,4]_{17}$  & $[[136,123,5]]_{17}$ \\
		\hline
		\makecell[c]{$(x+1)^{4}(x+2)$\\$(x+4)(x+8)^2$}&  	$(x+1)(x-1)$& $[136,126,6]_{17}$  &$[136,128,5]_{17}$  & $[[136,118,6]]_{17}$ \\
		\hline
	\makecell[c]{$(x+1)^{5}(x-1)$\\$(x+2)(x+4)(x+8)^{2}$}&  \makecell[c]{	$(x+1)(x+2)$\\$(x+8)(x-8)$}& $[136,122,7]_{17}$  &$[136,126,6]_{17}$  & $[[136,112,7]]_{17}$ \\
		\hline
		\makecell[c]{$(x+1)^{6}(x-1)(x+2)^2$\\$(x+4)(x+8)^3(x-8)$}&  	$(x+1)(x-2)$& $[136,120,8]_{17}$  &$[136,122,7]_{17}$  & $[[136,106,8]]_{17}$ \\
		\hline
	\end{tabular}
\end{table}

\end{example}
\begin{example}\label{exam6}
 Let $a=2,b=3$, then $q=2^ab+1=13$ and $p=13$. Let $r=3,s=1$, then $n=2^rp^s=104$. By factorization, we have $x^{104}-1=(x+1)^{13}(x-1)^{13}(x+5)^{13}(x-5)^{13}(x^2+5)^{13}(x^2-5)^{13}$. Put $C=\langle g(x)m(x) \rangle$ and $C^{'}=\langle g(x)\rangle$, where $g(x)m(x)$ is a factor of $x^{136}-1$. Specific parameters of some codes are shown in table \ref{tab6}.
\begin{table}[!h]  \label{tab6}
	\caption{\centering {Some QEC codes constructed from Conclusion (2) in Theorem \ref{theo12}} }
	\begin{tabular}{|c|c|c|c|c|}
		\hline
		$g(x)$ & $m(x)$&  $C$ &  $C^{'}$ &   QEC code \\
		\hline
		$(x+1)$&  	$(x+1)(x^2+5)$& $[104,100,3]_{17}$  &$[104,103,2]_{17}$  & $[[104,99,3]]_{17}$ \\
		\hline
	$(x+1)^2(x^2+5)$&  	$(x+1)(x^2-5)$& $[104,97,4]_{17}$  &$[104,100,3]_{17}$  & $[[104,93,4]]_{17}$ \\
		\hline
	$(x+1)^{3}(x^2+5)(x^2-5)$&  $(x+1)(x^2+5)$& $[104,93,5]_{17}$  &$[104,97,4]_{17}$  & $[[104,86,5]]_{17}$ \\
		\hline
	\makecell[c]{$(x+1)^{4}(x+5)$\\$(x^2+5)^2(x^2-5)$}&  $(x+1)$& $[104,92,6]_{17}$  &$[104,93,5]_{17}$  & $[[104,81,6]]_{17}$ \\
	\hline
		\makecell[c]{$(x+1)^{5}(x+5)$\\$(x^2+5)^2(x^2-5)$}& 	\makecell[c]{ $(x+1)(x-1)$\\$(x+5)(x^2+5)$}& $[104,87,7]_{17}$  &$[104,92,6]_{17}$  & $[[104,75,7]]_{17}$ \\
	\hline
		\makecell[c]{$(x+1)^{6}(x-1)$\\$(x+5)^2(x^2+5)^3(x^2-5)$}&  $(x+1)$& $[104,86,8]_{17}$  &$[104,87,7]_{17}$  & $[[104,69,8]]_{17}$ \\
	\hline
	\end{tabular}
\end{table}

\end{example}

\begin{example}\label{exam7}
 Let $a=2,b=3$, then $q=2^ab-1=11$ and $p=11$. Let $r=3,s=1$, then $n=2^rp^s=88$. By factorization, we have 
$x^{88}-1=(x+1)^{11}(x-1)^{11}(x^2+1)^{11}(x^2+3x+10)^{11}(x^2+8x+10)^{11}$. Put $C=\langle g(x)m(x) \rangle$ and $C^{'}=\langle g(x)\rangle$, where $g(x)m(x)$ is a factor of $x^{136}-1$. Specific parameters of some codes are shown in table \ref{tab7}.
\begin{table}[!h]  \label{tab7}
	\caption{\centering {Some QEC codes constructed from Conclusion (2) in Theorem \ref{theo13}} }
	\begin{tabular}{{|p{3.9cm}<{\centering}|p{2.1cm}<{\centering}|p{1.5cm}<{\centering}|p{1.5cm}<{\centering}|p{1.5cm}<{\centering}|}}
		\hline
		$g(x)$ & $m(x)$&  $C$ &  $C^{'}$ &   QEC code \\
		\hline
		$(x+1)$&  	\makecell[c]{$(x+1)$\\$(x^2+3x+10)$}& $[88,84,3]_{11}$  &$[88,87,2]_{11}$  & $[[88,82,3]]_{11}$ \\
		\hline
		$(x+1)^2(x^2+3x+10)$& $(x+1)$&  $[88,83,4]_{11}$  & $[88,84,3]_{11}$& $[[88,79,4]]_{11}$ \\
	\hline
	$(x+1)^3(x^2+3x+10)$& \makecell[c]{$(x+1)(x^2+1)$\\$(x^2+3x+10)$}&  $[88,78,5]_{11}$  & $[88,83,4]_{11}$& $[[88,73,5]]_{11}$ \\
	\hline
	$(x+1)^4(x^2+1)(x^2+3x+10)^2$&  $(x+1)$&  $[88,77,6]_{11}$  & $[88,78,5]_{11}$  & $[[88,67,6]]_{11}$ \\
	\hline
	\makecell[c]{$(x+1)^2(x-1)(x^2+1)$\\$(x^2+3x+10)^6(x^2+8x+10)$}&  $(x+1)$&  $[88,68,7]_{11}$  & $[88,69,6]_{11}$  & $[[88,49,7]]_{11}$ \\
	\hline
	\makecell[c]{$(x+1)^3(x-1)(x^2+1)$\\$(x^2+3x+10)^6(x^2+8x+10)$}&  $(x^2+3x+10)$&  $[88,66,8]_{11}$  & $[88,68,7]_{11}$  & $[[88,46,8]]_{11}$ \\
    \hline
\makecell[c]{$(x+1)^3(x-1)(x^2+1)$\\$(x^2+3x+10)^7(x^2+8x+10)$}&  \makecell[c]{$(x+1)(x^2+1)$\\$(x^2+3x+10)$}&  $[88,61,9]_{11}$  & $[88,66,8]_{11}$  & $[[88,39,9]]_{11}$ \\
\hline
\makecell[c]{$(x+1)^4(x-1)(x^2+1)^2$\\$(x^2+3x+10)^8(x^2+8x+10)$}&  \makecell[c]{$(x-1)^2(x^2+1)^2$\\$(x^2+3x+10)^2$}&  $[88,51,10]_{11}$  & $[88,61,9]_{11}$  & $[[88,24,10]]_{11}$ \\
\hline
\makecell[c]{$(x+1)^3(x-1)^3(x^2+1)^4$\\$(x^2+3x+10)^10(x^2+8x+10)$}&  $(x^2-1)(x^2+1)$&  $[88,48,11]_{11}$  & $[88,52,10]_{11}$  & $[[88,12,11]]_{11}$ \\
	\hline
	\end{tabular}
\end{table}
\end{example}

\begin{example}\label{exam8}
	Let $a=3,b=1$, then $q=2^ab-1=7$ and $p=7$. Let $r=3,s=1$, then $n=2^rp^s=56$. By factorization, we have 
$x^{56}-1=(x+1)^{7}(x-1)^{7}(x^2+1)^{7}(x^2+3x+1)^{7}(x^2+4x+1)^{7}$. Put $C=\langle g(x)m(x) \rangle$ and $C^{'}=\langle g(x)\rangle$, where $g(x)m(x)$ is a factor of $x^{136}-1$. Specific parameters of some codes are shown in table \ref{tab8}

\begin{table}[ht]  \label{tab8}
	\caption{\centering {Some QEC codes constructed from Conclusion (2) in Theorem \ref{theo13}} }
	\begin{tabular}{|c|c|c|c|c|}
		\hline
		$g(x)$ & $m(x)$&  $C$ &  $C^{'}$ &   QEC code\\
		\hline
		$(x+1)$&  	$(x+1)(x^2+3x+1)$& $[56,52,3]_{7}$  &$[56,55,2]_{7}$  & $[[56,51,3]]_{7}$ \\
		\hline
	$(x+1)^2(x^2+3x+1)$&  	$(x+1)$& $[56,51,4]_{7}$  &$[56,52,3]_{7}$  & $[[56,47,4]]_{7}$ \\
		\hline
	\end{tabular}
\end{table}
\end{example}
\begin{remark}\label{rem2}
The Hamming distances of all cyclic codes in Examples \ref{exam5},\ref{exam6},\ref{exam7} and \ref{exam8} can be determined using our latest research results in \cite{Pan2024}. Comparing with the results in Examples \ref{exam1},\ref{exam2},\ref{exam3},\ref{exam4}, it is easy to find that these QEC codes constructed in Examples \ref{exam5},\ref{exam6},\ref{exam7},\ref{exam8} have larger dimensions when the lengths and the minimum distances are the same, which indicates that these QEC codes have better parameters. Moreover, comparing these QEC codes in Examples \ref{exam5},\ref{exam6},\ref{exam7},\ref{exam8} and online table \cite{Grassl2007}, we find that our QEC codes are all new, meaning that they exhibit parameters distinct from those of previously known constructions. Specifically, the QEC code $[[104,86,5]]$ in Example \ref{exam2} is consistent with the known code constructed in \cite{Grassl2007}, and the others in Example \ref{exam5}, \ref{exam6} are better than all the QEC codes with the same length and Hamming distance given in \cite{Grassl2007}, that is, the dimension of our QEC codes is greater than the dimension of all the QEC codes with the same length and Hamming distance in \cite{Grassl2007}. In Example \ref{exam7}, our QEC codes $[[88,82,3]]$, $[[88,79,4]]$, $[[88,73,5]]$ and $[[88,67,6]]$ are consistent with some known codes constructed in \cite{Grassl2007}. In Example \ref{exam8}, our QEC codes $[[56,51,3]]$ and $[[56,47,4]]$ are consistent with some known codes constructed in \cite{Grassl2007}. Note that these known QEC codes in \cite{Grassl2007} and our QEC codes in Examples \ref{exam5},\ref{exam6},\ref{exam7},\ref{exam8} are constructed over different fields. This shows these QEC codes obtained in Examples \ref{exam5},\ref{exam6},\ref{exam7},\ref{exam8} are all new.
\end{remark}

\section{QEC MDS codes of length $2^rp^s$ }\label{sec5}

It is well-known that the parameters of a linear code must satisfy the Singleton bound: $d_H(C)\leq n-k+1$. 
In this section, we first construct all  MDS repeated-root cyclic codes with a length of $2^rp^s$ over $\mathbb{F}_q$. And then, based on CSS construction, we obtain all QEC MDS codes by using such MDS cyclic codes over $\mathbb{F}_q$.

Let $C$ be a repeated-root cyclic code of length $2^rp^s$ over $\mathbb{F}_q$. When $q= 2^ab+1$, from Theorem \ref{theo4}, we have $C= \textbf{\Big\langle}\textbf{\big(} m_0(x)\textbf{\big)}^{j_0}\textbf{\big(}m_{2^{r-1}}(x)\textbf{\big)}^{j_{2^{r-1}}}\prod\limits_{i=2}^r\prod\limits_{s\in D_i}\textbf{\big(}m_{s2^{r-i}}(x)\textbf{\big)}^{j_{s2^{r-i}}}\textbf{\Big\rangle}$,  where $0\leq j_h\leq p^s$ for each $h$. Suppose $a\geq r$, then $q\equiv 1 \pmod{2^r}$. In this case, it follows that the generator polynomial of $C$ is the product of some linear polynomials. Moreover, let $\omega$ be a primitive element of $\mathbb{F}_q$. Then $\theta=\omega^{\frac{q-1}{2^r}}\in\mathbb{F}_{q}^{*}$ is a $2^r$-th root of unity. Thus,  $x^{2^r}-1\in \mathbb{F}_q[x]$ can be decomposed into the product of $2^r$ irreducible polynomials as follows $$x^{2^r}-1=\prod_{i=0}^{2^r-1}(x-\theta^i).$$
For convenience, the generator polynomial of $C$ can be expressed as follows
$$g(x)=\prod_{h=0}^{2^r-1}(x-\theta^{i_h})^{j_h},$$
where $\{i_0,i_1,\cdots,i_{2^r-1}\}$ is a permutation of $\{0,1,2,\cdots,2^r-1\}$ such that $0\leq j_{2^r-1}\leq \cdots\leq j_{1}\leq j_{0}\leq p^s$. Thus, $C=\textbf{\Big\langle}\prod\limits_{h=0}^{2^r-1}(x-\theta^{i_h})^{j_h}\textbf{\Big\rangle}$ and it has $p^{m(2^rp^s-\sum_{h=0}^{2^r-1}j_h)}$ codewords. Note that $(x-\theta^i)^{*}=x-\theta^{2^r-i}$ for any $i\in \{0,1,2,\cdots,2^r-1\}$. The dual code of $C$ is the cyclic code $C^{\perp}=\textbf{\Big\langle}\prod\limits_{t=0}^{2^r-1}(x-\theta^{i_t})^{p^s-j_{2^r-i_t}}\textbf{\Big\rangle}$. 

For any $T=\{0,1,2,\cdots,p^s-1\}$ and $P_t=\prod^{s-1}\limits_{i=0}(t_i+1)$, where $t\in T$ and $t_i$ is the coefficient the $p$-adic representation of $t$. Let $1\leq\tau\leq p-2$ be an integer. From \cite{O2009}, the set $T$ can be represented as a union of $s$ disjoint sets as follows:
 $$T=\{0\}\cup\bigcup_{k=0}^{s-1}\bigcup_{\tau=0}^{p-2}\{i:p^s-p^{s-k}+\tau p^{s-k-1}+1\leq i\leq p^s-p^{s-k}+(\tau+1) p^{s-k-1}\}.$$
For any nonzero integer $l\in T$, there must exist integers $\tau$ and $k$ such that $p^s-p^{s-k}+\tau p^{s-k-1}+1\leq l\leq p^s-p^{s-k}+(\tau+1) p^{s-k-1}$. It follows from \cite{O2009} that $\textrm{min}\{P_t:t\geq l ~ \textrm{and }~t \in T\}=(\tau+2)p^{k}$. For fix $t\in T$, let us denote $g_t(x)=\prod_{h=0}^{2^r-1}(x-\theta^{i_h})^{j_{h,t}}$, where $e_{h,t}=1$ if $j_h>t$, otherwise $j_{h,t}=0$. Then $C_{t}=\langle g_t(x)\rangle$ is a cyclic code of length $2^r$ over $F_q$. According to the relationship between the minimum distance of repeated-root cyclic codes and single-root cyclic codes \cite{O2009}, we can obtain the following crucial result
\begin{eqnarray} \label{equ 2}
	d_H(C)=\textrm{min}\{P_td_H(C_t):t\in T\}.
\end{eqnarray} 
\begin{theorem}\label{theo14}
Let $r$ be an integer and $q= 2^ab+1$ with $a\geq r
$. Suppose that $C=\textbf{\Big\langle}\prod\limits_{h=0}^{2^r-1}(x-\theta^{i_h})^{j_h}\textbf{\Big\rangle}$ is a repeated-root cyclic code of length $2^rp^s$ over $\mathbb{F}_q$, where $\{i_1,i_2,\cdots,i_{2^r-1}\}$ is a permutation of $\{0,1,2,\cdots,2^r-1\}$ such that $0\leq j_{2^r-1}\leq \cdots\leq j_{1}\leq j_{0}\leq p^s$. When $0<j_h<p^s$, there is $t_h=p^s-p^{s-k_h}+(\tau_h+1) p^{s-k_h-1}+1$ such that $p^s-p^{s-k_h}+\tau_h p^{s-k_h-1}+1\leq j_h\leq t_h$. Denote the numbers of $j_h=p^s$ by $N$. Then 
	$$d_H(C)=\left\{ {{\begin{array}{ll}
				{\leq 2^r}, & \textrm{if}~ j_{2^r-1}=0,\\
				{ \textrm{min}\{(\tau_h+2)p^{k_h}d_H(C_{t_h}):h=N,N+1,\dots,2^r-1\}}, & \textrm{if}~ j_{2^r-1}>0.\\
	\end{array} }} \right.$$
\end{theorem}
\begin{proof} 
If $j_{2^s-1}=0$, we know that the  generator polynomial $g(x)=\prod_{t=0}^{2^r-2}(x-\theta^{i_h})^{j_h}$  must be a factor of $\prod_{h=0}^{2^r-2}(x-\theta^{i_h})^{p^s}$, which has at most $2^r$ non-zero terms. Therefore, we have $1\leq d_{H}(C)\leq 2^r$. If $j_{2^r-1}>0$, it follows from (\ref{equ 2}) that this result is obvious.
\end{proof} 

In the following, we can get all the MDS cyclic codes of length $2^rp^s$.

\begin{theorem}\label{theo15}
Let $r$ be an integer, $q= 2^ab+1$ with $a\geq r
$. Suppose that $C=\textbf{\Big\langle}\prod\limits_{h=0}^{2^r-1}(x-\theta^{i_h})^{j_h}\textbf{\Big\rangle}$ is a repeated-root cyclic code of length $2^rp^s$ over $\mathbb{F}_q$, where $\{i_1,i_2,\cdots,i_{2^r-1}\}$ is a permutation of $\{0,1,2,\cdots,2^r-1\}$ such that $0\leq j_{2^r-1}\leq \cdots\leq j_{1}\leq j_{0}\leq p^s$. Then $C$ is an MDS code iff one of the following conditions is satisfied.\\
	(1) $j_0=0$. In this case, $d_{H}(C)=1$.\\
	(2) $j_0=1, j_1=0$. In this case, $d_{H}(C)=2$.\\iff
	(3) $j_0=j_1=\cdots =j_{2^r-2}=p^s, j_{2^r-1}=p^s-1$. In this case, $d_{H}(C)=2^rp^s$. 
\end{theorem}

\begin{proof} Obviously, the cyclic code $C$ has dimension $k=2^rp^s-\sum_{h=0}^{2^r-1}j_h$. Furthermore, it is known that $C$ is an MDS code iff $d_H(C)=2^rp^s-k+1=\sum_{h=0}^{2^r-1}j_h+1$. Now, we can offer the following two cases are discussed.
		
\textbf{Case 1}: $j_{2^r-1}=0$. In this case, from Theorem $\ref{theo14}$, we have $d_H(C)\leq 2^r$. Next, we will analyze the value of $j_0$ on a case-by-case basis.

If $j_0=0$, we then have $C=\langle1\rangle$ and $d_H(C)=1$. Since $0\leq j_{2^r-1}\leq \cdots\leq j_{1}\leq j_{0}\leq p^s$, we know that $j_{1}=j_2=\cdots=j_{2^r-1}=0$. Hence, $d_H(C)=\sum_{h=0}^{2^r-1}j_t+1$ always holds.	
	
If $0<j_0\leq p^{s-1}$, then $j_h \leq p^{s-1}$ for any $h\in\{1,2,\cdots,2^r-1\}$. Obviously, we have $\prod_{h=0}^{2^r-1}(x-\theta^{i_h})^{j_h} |(x^{2^r}-1)^{p^{s-1}}$. This implies that $d_H(C)\leq2$. Hence, $d_H(C)=\sum_{h=0}^{2^r-1}j_h+1$ holds iff $j_0=1$ and the remaining $j_h=0$. 
		
If $p^{s-1}+1\leq j_0\leq 2p^{s-1}$, then $j_h\leq 2p^{s-1}$ for any $h\in\{1,2,\cdots,2^r-1\}$. Denote $t=2p^{s-1}$, it is evident that $P_t=3$ and $d_H(C_t)=1$. From (\ref{equ 2}), $d_H(C)\leq p_td_H(C_t)=3$. Since $j_0\geq 2$, then $\sum_{h=0}^{2^r-1}j_h+1\geq 3$. Hence, $d_H(C)=\sum_{h=0}^{2^r-1}j_h+1$ holds iff $d_H(C)=\sum_{h=0}^{2^r-1}j_h+1=3$ iff $j_0=2$ and the remaining $j_h=0$. This contradicts the fact that $d_H(C)=2$ when $j_0=2$ and the remaining $j_h=0$.
	
More generally, if $\lambda p^{s-1}+1\leq j_0\leq(\lambda+1) p^{s-1}$ with $2\leq\lambda\leq p-2$, then $j_h\leq (\lambda+1)p^{s-1}$ for any $h\in\{1,2,\cdots,2^r-1\}$. Denote $t=(\lambda+1)p^{s-1}$, it is evident that $P_t=\lambda+2$ and $d_H(C_t)=1$. From (\ref{equ 2}), $d_H(C)\leq p_td_H(C_t)=\lambda+2$. Since $j_0\geq \lambda+1$, then $\sum_{h=0}^{2^r-1}j_h+1\geq \lambda+2$. Hence, $d_H(C)=\sum_{h=0}^{2^r-1}j_h+1$ holds iff $d_H(C)=\sum_{h=0}^{2^r-1}j_h+1=\lambda+2$ iff $j_0=\lambda+1$ and the remaining $j_h=0$. This contradicts the fact that $d_H(C)=2$ when $j_0=\lambda+1$ and the remaining $j_h=0$.	
			
If $(p-1)p^{s-1}+1\leq j_0\leq p^{s}$, recall that $\sum_{h=0}^{2^r-1}j_h+1\geq j_0+1$ and $d_H(C)\leq 2^r$. Then $d_H(C)=\sum_{h=0}^{2^r-1}j_h+1$ is impossible when $2^r \leq j_0$. Therefore, we only need to discuss the case of $j_0<2^r$.
	
When $j_0<2^r\leq p^s$, we have $j_1< p^{s-1}-1$, otherwise, $\sum_{h=0}^{2^r-1}j_h+1>j_0+j_1\geq p^s$ which means that $d_H(C)=\sum_{h=0}^{2^r-1}j_h+1$ is impossible. Denote $t=p^{s-1}$, it is evident that $P_t=2$ and $d_H(C_t)=2$. From (\ref{equ 2}), $d_H(C)\leq p_td_H(C_t)=4$. 
Note that $\sum_{h=0}^{2^r-1}j_h+1\geq j_0+1\geq (p-1)p^{s-1}+2$ and $(p-1)p^{s-1}+2$ takes the minimum value $4$ only if $p=3$ and $s=1$. Hence,  $d_H(C)=\sum_{h=0}^{2^r-1}j_h+1$ holds iff $d_H(C)=\sum_{h=0}^{2^r-1}j_h+1=4$ iff $j_0=(p-1)p^{s-1}+1$ with $p=3,s=1$ and  the remaining $j_h=0$. This contradicts the fact that $d_H(C)=2$ when $j_0=(p-1)p^{s-1}+1$ and the remaining $j_h=0$.  
	
When $p^s<2^r\leq2p^s$, more precisely, we have $p^s<2^r<2p^s$ as $p$ is odd.  Now, we first consider the case of $s=1$. At this time, we have $j_0=p$ and $p<2^r<2p$. If $j_1=p$, then $\sum_{h=0}^{2^r-1}j_h+1\geq 2p+1$. As $d_H(C)\leq 2^r<2p$, $d_H(C)=\sum_{h=0}^{2^r-1}j_h+1$ is impossible. If $j_1=p-1$, then $\sum_{h=0}^{2^r-1}j_h+1\geq 2p$. Since $d_H(C)\leq 2^r<2p$, $d_H(C)=\sum_{h=0}^{2^r-1}j_h+1$ is impossible. Next, if $0\leq j_1\leq p-2$, denote $t=j_1$, then  it is evident that $P_t=j_1+1$ and $d_H(C_t)=2$. From (\ref{equ 2}), we have $d_H(C)\leq p_td_H(C_t)=2(j_1+1)$. Note that $ p+j_1+1>2(j_1+1)$ always holds when $0\leq j_1\leq p-2$. Hence, $\sum_{h=0}^{2^r-1}j_h+1\geq p+j_1+1>d_H(C)$, which means $d_H(C)=\sum_{h=0}^{2^r-1}j_h+1$ is impossible. Second, we consider the case of $s\geq 2$. Furthermore, we have $j_2< (p-1)p^{s-1}-1$, otherwise, $\sum_{h=0}^{2^r-1}j_h+1\geq j_0+j_1+j_2+1\geq 3(p-1)p^{s-1}=2p^s+(p-3)p^{s-1}\geq 2p^s$ which means that $d_H(C)=\sum_{h=0}^{2^r-1}j_h+1$ is impossible. Denote $t=(p-1)p^{s-1}$, it is evident that $P_t=p$ and $d_H(C_t)=2$ or $3$. If $d_H(C_t)=2$, it follows from (\ref{equ 2}) that $d_H(C)\leq p_td_H(C_t)=2p$. 
Note that $\sum_{h=0}^{2^r-1}j_h+1\geq j_0+1\geq (p-1)p^{s-1}+2>2p$. At this time, $d_H(C)=\sum_{h=0}^{2^r-1}j_h+1$ is impossible. Moreover, if $d_H(C_t)=3$, it follows from (\ref{equ 2}) that $d_H(C)\leq p_td_H(C_t)=3p$. In fact, we also know that $j_1\geq(p-1)p^{s-1}+1$ if $d_H(C_t)=3$ with $t=(p-1)p^{s-1}$. Then, $\sum_{h=0}^{2^r-1}j_h+1\geq j_0+j_1+1\geq 2(p-1)p^{s-1}+3>3p$. Hence, $d_H(C)=\sum_{h=0}^{2^r-1}j_h+1$ is impossible.
	
Similarly, we can prove that $d_H(C)=\sum_{h=0}^{2^r-1}j_h+1$ is impossible when $2p^s<2^r\leq3p^s$ or $3p^s<2^r\leq4p^s\cdots $. These process is similar to the discussion of the case of $p^s<2^r\leq2p^s$, so we omit it here.
	
\textbf{Case 2}: $j_{2^r-1}>0$. From Theorem $\ref{theo14}$, we have $d_H(C)=\textrm{min}\{(\tau_h+2)p^{k_h}d_H(C_{t_h}):h=N,N+1,\dots,2^r-1\}$, where $N$ is the numbers of $j_h=p^s$. Since $0<j_N<p^s-1$, there is $t_N=p^s-p^{s-k_N}+(\tau_N+1) p^{s-k_N-1}+1$ such that $p^s-p^{s-k_N}+\tau_N p^{s-k_N-1}+1\leq j_N\leq t_N$. Put $t=t_N$, we have $2\leq d_H(C_t)\leq N+1$. When $0\leq N\leq 2^r-2$, by calculations, we have 
\begin{eqnarray*}
		\sum_{h=0}^{2^r-1}j_h
		&&\geq \underbrace{p^s+\cdots +p^s}_N+\sum_{h=N}^{2^r-1}(p^s-p^{s-k_h}+\tau_h p^{s-k_h-1}+1)\\
		&&\geq Np^s+(p^s-p^{s-k_N}+\tau_N p^{s-k_N-1}+1)+(2^r-N-1)\\
		&&\geq(N+1)(p^s-p^{s-k_N}+\tau_N p^{s-k_N-1}+1)+2^r-1\\
		&&=(N+1)(p^{s-k_N}(p^{k_N}-1)+\tau_N p^{s-k_N-1}+1)+2^r-1 \\
		&&\geq(N+1)(p(p^{k_N}-1)+\tau_N +1)+2^r-1 \\
		&&\geq(N+1)((\tau_N+2)(p^{k_N}-1)+\tau_N +1)+2^r-1 \\
		&&= (N+1)((\tau_N+2)p^{k_N}-1)+2^r-1\\
		&&=(N+1)(\tau_N+2)p^{k_N}+2^r-N-2\\
		&&>(N+1)(\tau_N+2)p^{k_N}-1\\
		&&\geq(\tau_N+2)p^{k_N}d_H(C_{t})-1\\
		&&\geq\textrm{min}\{(\tau_h+2)p^{k_h}d_H(C_{t_h}):h=N,\dots,2^r-1\}-1\\
		&&=d_H(C)-1.
\end{eqnarray*}
Hence, $d_H(C)=\sum_{h=0}^{2^r-1}j_h+1$ is impossible. That means that $C$ is not MDS. 

When $N=2^r-1$, by calculations, we have
	\begin{eqnarray*}
	\sum_{h=0}^{2^r-1}j_h
	&&\geq \underbrace{p^s+\cdots +p^s}_{2^r-1}+p^s-p^{s-k_{2^r-1}}+\tau_{2^r-1} p^{s-k_{2^r-1}-1}+1\\
	&&=(2^r-1)p^s+(p^s-p^{s-k_{2^r-1}}+\tau_{2^r-1} p^{s-k_{2^r-1}-1}+1)\\
	&&\geq2^r(p^s-p^{s-k_{2^r-1}}+\tau_N p^{s-k_{2^r-1}-1}+1)+2^r-1\\
	&&\geq 2^r((\tau_{2^r-1}+2)p^{k_{2^r-1}}-1)+2^r-1\\
	&&=2^r(\tau_{2^r-1}+2)p^{k_{2^r-1}}-1\\
	&&\geq(\tau_{2^r-1}+2)p^{k_{2^r-1}}d_H(C_{t})-1\\
	&&\geq\textrm{min}\{(\tau_h+2)p^{k_h}d_H(C_{t_h}):h=2^r-1\}-1\\
	&&=d_H(C)-1.
\end{eqnarray*}	
Hence, $d_H(C)=\sum_{h=0}^{2^r-1}j_h+1$ holds iff $d_H(C)=\sum_{h=0}^{2^r-1}j_h+1=2^r(\tau_{2^r-1}+2)p^{k_{2^r-1}}$ iff $\tau_{2^r-1}=p-2$ and $k_{2^r-1}=s-1$, i.e.,  $j_0=j_1=\cdots =j_{2^r-2}=p^s, j_{2^r-1}=p^s-1$. 

When $N=2^r$, i.e., $j_0=j_1=\cdots =j_{2^r-1}=p^s$, it is obvious that the cyclic code $C$ is not MDS.
\end{proof}
Let $C=\langle g(x)\rangle$ be a repeated-root cyclic code of length $2^rp^s$ over $\mathbb{F}_q$. When $q= 2^ab+1$ and $a<r$, according to  Theorem \ref{theo4}, we have 
$$C= \textbf{\Big\langle}\textbf{\big(} m_0(x)\textbf{\big)}^{j_0}\textbf{\big(}m_{2^{r-1}}(x)\textbf{\big)}^{j_{2^{r-1}}}\prod_{i=2}^r\prod_{s\in D_i}\textbf{\big(}m_{s2^{r-i}}(x)\textbf{\big)}^{j_{s2^{r-i}}}\textbf{\Big\rangle},$$ where $0\leq j_h\leq p^s$ for each $h$. According to  Theorem \ref{theo9}, we know that the degree of $m_{s2^{r-i}}(x)$ is $1$ for any $s\in D_i$ with $2\leq i\leq a$ and the degree of $m_{s2^{r-i}}(x)$ is $2^{i-a}$ for any $s\in D_i$ with $a+1\leq i\leq r$. It is known that $\textrm{deg}(m_0(x))=\textrm{deg}(m_{2^{r-1}}(x))=1$. So, the cyclic code $C$ has dimension 
$$k=2^rp^s-(j_0+j_{2^{r-1}}+\sum_{i=2}^{a}\sum_{s\in D_i}j_{s2^{r-i}}+\sum_{i=a+1}^{r}\sum_{s\in D_i}(2^{i-a}j_{s2^{r-i}})).$$ 
For the purpose of facilitating discussion, we will denote the cyclic code $C$ in the following form: 
{\small$$C= \textbf{\Big\langle}\textbf{\big(} m_0(x)\textbf{\big)}^{j_0}\textbf{\big(}m_{2^{r-1}}(x)\textbf{\big)}^{j_{2^{r-1}}}\prod_{i=2}^a\prod_{s\in D_i}\textbf{\big(}m_{s2^{r-i}}(x)\textbf{\big)}^{j_{s2^{r-i}}}\prod_{i=a+1}^r\prod_{s\in D_i}\textbf{\big(}m_{s2^{r-i}}(x)\textbf{\big)}^{j_{s2^{r-i}}}\textbf{\Big\rangle}.$$}
Furthermore, we can get the following theorem.
\begin{theorem}\label{theo16}
	Let $r$ be an integer, $q= 2^ab+1$ with $a< r$. Suppose that $C= \textbf{\Big\langle}\textbf{\big(} m_0(x)\textbf{\big)}^{j_0}\textbf{\big(}m_{2^{r-1}}(x)\textbf{\big)}^{j_{2^{r-1}}}\prod\limits_{i=2}^a\prod\limits_{s\in D_i}\textbf{\big(}m_{s2^{r-i}}(x)\textbf{\big)}^{j_{s2^{r-i}}}\prod\limits_{i=a+1}^r\prod\limits_{s\in D_i}\textbf{\big(}m_{s2^{r-i}}(x)\textbf{\big)}^{j_{s2^{r-i}}}\textbf{\Big\rangle}$
	is a repeated-root cyclic code of length $2^rp^s$ over $\mathbb{F}_q$, where $0\leq j_{h}\leq p^s$ for each $h$. Then $C$ is an MDS code iff one of the following conditions is satisfied.\\
	(1) $j_h=0$ for each $h$. In this case, $d_{H}(C)=1$.\\
	(2) There is only one $j_h\in \{j_0,j_{2^{r-1}}\}\cup\{j_{s2^{r-i}}:s\in D_i, 2\leq i\leq a\}$ such that $j_h=1$ and the remaining $j_h=0$. In this case, $d_{H}(C)=2$.\\
	(3) There is only one $j_h\in \{j_0,j_{2^{r-1}}\}\cup\{j_{s2^{r-i}}:s\in D_i, 2\leq i\leq a\}$ such that $j_h=p^s-1$ and the remaining $j_h=p^s$. In this case, $d_{H}(C)=2^rp^s$. 
\end{theorem}
\begin{proof} It is known that $C$ is an MDS code iff $d_H(C)=2^rp^s-k+1=j_0+j_{2^{r-1}}+\sum_{i=2}^{a}\sum_{s\in D_i}j_{s2^{r-i}}+\sum_{i=a+1}^{r}\sum_{s\in D_i}(2^{i-a}j_{s2^{r-i}})+1$. Denote $J=\textrm{max}\{\{j_0,j_{2^{r-1}}\}\cup\{j_{s2^{r-i}}:s\in D_i, 2\leq i\leq a\}\}$, $J_1=\textrm{max}\{\{j_0,j_{2^{r-1}}\}\cup\{j_{s2^{r-i}}:s\in D_i, 2\leq i\leq a\}/\{J\}\}$, $J_2=\textrm{max}\{\{j_0,j_{2^{r-1}}\}\cup\{j_{s2^{r-i}}:s\in D_i, 2\leq i\leq a\}/\{J,J_1\}\}$ and $j=\textrm{min}\{\{j_0,j_{2^{r-1}}\}\cup\{j_{s2^{r-i}}:s\in D_i, 2\leq i\leq a\}\}$. Furthermore, we can consider the following two cases: $j=0$ and $j>0$. 

\textbf{Case 1}: $j=0$. In a similar way to Theorem $\ref{theo14}$, we can get $d_H(C)\leq 2^r$. In the following, we will analyze the value of $J$ on a case-by-case basis.
	
If $J=0$, we then have $C=\langle1\rangle$ and $d_H(C)=1$. Hence, $d_H(C)=j_0+j_{2^{r-1}}+\sum_{i=2}^{a}\sum_{s\in D_i}j_{s2^{r-i}}+\sum_{i=a+1}^{r}\sum_{s\in D_i}(2^{i-a}j_{s2^{r-i}})+1$ always holds.	
	
If $0<J\leq p^{s-1}$, then $j_h \leq p^{s-1}$ for any $h$. Obviously, we have $g(x) |(x^{2^r}-1)^{p^{s-1}}$. This implies that $d_H(C)\leq2$. Hence, $d_H(C)=j_0+j_{2^{r-1}}+\sum_{i=2}^{a}\sum_{s\in D_i}j_{s2^{r-i}}+\sum_{i=a+1}^{r}\sum_{s\in D_i}(2^{i-a}j_{s2^{r-i}})+1$ holds iff $d_H(C)=2$, $J=1$ and the remaining $j_h=0$ iff there exist only one $j_h\in \{j_0,j_{2^{r-1}}\}\cup\{j_{s2^{r-i}}:s\in D_i, 2\leq i\leq a\}$ such that $j_h=1$ and the remaining $j_h=0$.

Moreover, for the subcase of $\lambda p^{s-1}+1\leq J \leq(\lambda+1) p^{s-1}$ with $1\leq\lambda\leq p-2$ and the subcase of $(p-1)p^{s-1}+1\leq J\leq p^{s}$, the proof is identical to that of theorem $\ref{theo15}$, and thus will not be reiterated here.

\textbf{Case 2}: $j>0$. The proof is also similar to that of theorem $\ref{theo15}$. Thus, we omit it here.
\end{proof}	

Combining Theorem \ref{theo5} and Proposition \ref{prop7}, we can construct the following QEC MDS codes by using the above MDS cyclic codes of length $2^rp^s$ over $\mathbb{F}_q$.

\begin{theorem}\label{theo17}
	Let $r$ be an integer, $q=2^ab+1$ with $a\geq r
	$. Suppose that $C= \textbf{\Big\langle}\textbf{\big(} m_0(x)\textbf{\big)}^{j_0}\textbf{\big(}m_{2^{r-1}}(x)\textbf{\big)}^{j_{2^{r-1}}}\prod\limits_{i=2}^r\prod\limits_{s\in D_i}\textbf{\big(}m_{s2^{r-i}}(x)\textbf{\big)}^{j_{s2^{r-i}}}\textbf{\Big\rangle}$ is a repeated-root cyclic code of length $2^rp^s$ over $\mathbb{F}_q$, where $0\leq j_{h}\leq p^s$ for each $h$. Denote $D=\{j_0,j_{2^{r-1}}\}\cup\{j_{s2^{r-i}}:s\in D_i, 2\leq i\leq r\}$. Then we have following results.
	
	(1) if $j_h=0$ for each $h$, then a QEC MDS code with parameters $[[2^rp^s,2^rp^s,1]]_q$ is obtained.
	
	(2) if $\sum\limits_{h\in D}j_{h}=1$, then a QEC MDS code with parameters $[[2^rp^s,2^rp^s-2,2]]_q$ is obtained.
\end{theorem}

\begin{proof} 
(1) If $j_h=0$ for each $h$, from Theorem \ref{theo15}, the cyclic code $C$ is a $[2^rp^s,2^rp^s,1]_q$ MDS code. From Proposition \ref{prop7}, we know that $C^{\bot}\subseteq C$ if $j_h=0$ for each $h$. Based on CSS construction, we can obtain a QEC MDS code with parameters $[[2^rp^s,2^rp^s,1]]_q$.

(2) if $\sum_{h\in D}j_{h}=1$, from Theorem \ref{theo15}, the cyclic code $C$ is a $[2^rp^s,2^rp^s-1,2]_q$ MDS code. From Proposition \ref{prop7}, we know that $C^{\bot}\subseteq C$. Based on CSS construction, we can obtain a QEC MDS code with parameters $[[2^rp^s,2^rp^s-2,2]]_q$.
\end{proof} 

\begin{theorem}\label{theo18}
	Let $r$ be an integer, $q=2^ab+1$ with $a<r
	$. Suppose that $C= \textbf{\Big\langle}\textbf{\big(} m_0(x)\textbf{\big)}^{j_0}\textbf{\big(}m_{2^{r-1}}(x)\textbf{\big)}^{j_{2^{r-1}}}\prod\limits_{i=2}^a\prod\limits_{s\in D_i}\textbf{\big(}m_{s2^{r-i}}(x)\textbf{\big)}^{j_{s2^{r-i}}}\prod\limits_{i=a+1}^r\prod\limits_{s\in D_i}\textbf{\big(}m_{s2^{r-i}}(x)\textbf{\big)}^{j_{s2^{r-i}}}\textbf{\Big\rangle}$ is a repeated-root cyclic code of length $2^rp^s$ over $\mathbb{F}_q$, where $0\leq j_{h}\leq p^s$ for each $h$. Denote $D=\{j_0,j_{2^{r-1}}\}\cup\{j_{s2^{r-i}}:s\in D_i, 2\leq i\leq a\}$. Then we have following results.
	
	(1) if $j_h=0$ for each $h$, then a QEC MDS code with parameters $[[2^rp^s,2^rp^s,1]]_q$ is obtained.
	
	(2) if $\sum\limits_{h\in D} j_{h}=1$ and the remaining $j_h=0$, then a QEC MDS code with parameters $[[2^rp^s,2^rp^s-2,2]]_q$ is obtained.
\end{theorem}
\begin{example}\label{exam9}
 Let $a=3,b=3$, then $q=2^ab+1=25$ and $p=5$. Let $r=3,s=1$, then $n=2^rp^s=40$. Denote $\theta=\omega^3$, then we have $x^{40}-1=\prod_{h=0}^{7}(x-\theta^h)^{5}$. Put $C=\langle\prod_{h=0}^{7}(x-\theta^h)^{j_h}\rangle$, where $0\leq j_h\leq5$. 

(i) If $j_0=j_1=\cdots=j_7=0$, then $C=\langle1\rangle$ and $d_H(C)=1$. Obviously, the cyclic code $C$ is a $[40,40,1]_{25}$ MDS code. It follows from Proposition \ref{prop8} that $C^{\bot}\subseteq C$. From Theorem \ref{theo17}, we obtain a $[[40,40,1]]_{25}$ QEC MDS code.

(ii) If $\sum\limits_{h=0}^7 j_{h}=1$, then $C=\langle x-\theta^h\rangle$ with some $0\leq h\leq 7$ and $d_H(C)=2$. Obviously, the cyclic code $C$ is a $[40,39,2]_{25}$ MDS code. It follows from Proposition \ref{prop8} that $C^{\bot}\subseteq C$. From Theorem \ref{theo17}, we obtain a $[[40,38,2]]_{25}$ QEC MDS code.
\end{example}

\begin{example}\label{exam10}
 Let $a=2,b=7$, then $q=2^ab+1=29$ and $p=29$. Let $r=3,s=1$, then $n=2^rp^s=232$. By factorization, we have $x^{232}-1=(x+1)^{29}(x+12)^{29}(x+17)^{29}(x+28)^{29}(x^2+12)^{29}(x^2+17)^{29}$. Put $C=(x+1)^{j_0}(x+12)^{j_1}(x+17)^{j_2}(x+28)^{j_3}(x^2+12)^{j_4}(x^2+17)^{j_5}\rangle$, where $0\leq j_h\leq29$. 

(i) If $j_0=j_1=\cdots=j_5=0$, then $C=\langle1\rangle$ and $d_H(C)=1$. Obviously, the cyclic code $C$ is a $[232,232,1]_{29}$ MDS code. It follows from Proposition \ref{prop8} that $C^{\bot}\subseteq C$. From Theorem \ref{theo17}, we obtain a $[[232,232,1]]_{29}$ QEC MDS code.

(ii) If $\sum\limits_{h=0}^3 j_{h}=1$ and $j_4=j_5=0$, then $C=\langle x+1\rangle$ or $\langle x+12\rangle$ or $\langle x+17\rangle$ or $\langle x+28\rangle$ and $d_H(C)=2$. Obviously, the cyclic code $C$ is a $[232,231,2]_{29}$ MDS code. It follows from Proposition \ref{prop8} that $C^{\bot}\subseteq C$. From Theorem \ref{theo17}, we obtain a $[[232,230,2]]_{29}$ QEC MDS code.
\end{example}

Let $C$ be a repeated-root cyclic code of length $2^rp^s$ over $\mathbb{F}_q$. When $q= 2^ab-1$, from Theorem \ref{theo4}, we have $C= \textbf{\Big\langle}\textbf{\big(} m_0(x)\textbf{\big)}^{j_0}\textbf{\big(}m_{2^{r-1}}(x)\textbf{\big)}^{j_{2^{r-1}}}\textbf{\big(}m_{2^{r-2}}(x)\textbf{\big)}^{j_{2^{r-2}}}\prod\limits_{i=3}^r\prod\limits_{s\in O_i}\textbf{\big(}m_{s2^{r-i}}(x)\textbf{\big)}^{j_{s2^{r-i}}}\textbf{\Big\rangle}$,  where $0\leq j_h\leq p^s$ for each $h$.  Now, if we consider the cyclic code $C$ over $\mathbb{F}_q$, the following theorem is obtained, and the proof process is similar to Theorem \ref{theo15}.

\begin{theorem}\label{theo19}
	Let $q$ be a prime power and has the form $q=2^ab-1$. Suppose that $C= \textbf{\Big\langle}\textbf{\big(} m_0(x)\textbf{\big)}^{j_0}\textbf{\big(}m_{2^{r-1}}(x)\textbf{\big)}^{j_{2^{r-1}}}\textbf{\big(}m_{2^{r-2}}(x)\textbf{\big)}^{j_{2^{r-2}}}\prod\limits_{i=3}^r\prod\limits_{s\in O_i}\textbf{\big(}m_{s2^{r-i}}(x)\textbf{\big)}^{j_{s2^{r-i}}}\textbf{\Big\rangle}$ is a repeated-root cyclic code of length $2^rp^s$ over $\mathbb{F}_q$, where $0\leq j_{h}\leq p^s$ for each $h$. Then $C$ is an MDS codes over $\mathbb{F}_q$ iff one of the following conditions is satisfied.\\
	(1) $j_h=0$ for each $h$. In this case, $d_{H}(C)=1$.\\
(2) There is only one $j_h\in \{j_0,j_{2^{r-1}}\}$ such that $j_h=1$,  and the remaining $j_h=0$. In this case, $d_{H}(C)=2$.\\
(3) There is only one $j_h\in \{j_0,j_{2^{r-1}}\}$ such that $j_h=p^s-1$ and the remaining $j_h=p^s$. In this case, $d_{H}(C)=2^rp^s$. 
\end{theorem}

Combining Proposition \ref{prop8} and Theorem \ref{theo19}, we can obtain the following theorem based on CSS construction.

\begin{theorem}\label{theo20}
		Let $q$ be a prime power and has the form $q=2^ab-1$. Suppose that $C= \textbf{\Big\langle}\textbf{\big(} m_0(x)\textbf{\big)}^{j_0}\textbf{\big(}m_{2^{r-1}}(x)\textbf{\big)}^{j_{2^{r-1}}}\textbf{\big(}m_{2^{r-2}}(x)\textbf{\big)}^{j_{2^{r-2}}}\prod\limits_{i=3}^r\prod\limits_{s\in O_i}\textbf{\big(}m_{s2^{r-i}}(x)\textbf{\big)}^{j_{s2^{r-i}}}\textbf{\Big\rangle}$ is a repeated-root cyclic of length $2^rp^s$ over $\mathbb{F}_q$, where $0\leq j_{h}\leq p^s$ for each $h$. Then we have the following results.
	
	(1) if $j_h=0$ for each $h$, then a QEC MDS code with parameters $[[2^rp^s,2^rp^s,1]]$ is obtained.
	
	(2) if there is only one $j_h\in \{j_0,j_{2^{r-1}}\}$ such that $j_h=1$ and the remaining $j_h=0$, then a QEC MDS code with parameters $[[2^rp^s,2^rp^s-2,2]]$ is obtained.
\end{theorem}
\begin{example}\label{exam11}
 Let $a=3,b=3$, then $q=2^ab-1=23$ and $p=23$. Let $r=3,s=1$, then $n=2^rp^s=184$. By factorization, we have $x^{184}-1=(x+1)^{23}(x+22)^{23}(x^2+1)^{23}(x^2+5x+1)^{23}(x^2+18x+1)^{23}$. Put $C=\langle(x+1)^{j_0}(x+22)^{j_1}(x^2+1)^{j_2}(x^2+5x+1)^{j_3}(x^2+18x+1)^{j_4}\rangle$, where $0\leq j_h\leq23$. 

(i) If $j_0=j_1=\cdots=j_4=0$, then $C=\langle1\rangle$ and $d_H(C)=1$. Obviously, the cyclic code $C$ is a $[184,184,1]_{23}$ MDS code. It follows from Proposition \ref{prop8} that $C^{\bot}\subseteq C$. From Theorem \ref{theo17}, we obtain a $[[184,184,1]]_{25}$ QEC MDS code.

(ii) If $\sum_{h=0}^1 j_{h}=1$ and $j_2=j_3=j_4=0$, then $C=\langle x+1\rangle$ or $\langle x+22\rangle$ and $d_H(C)=2$. Obviously, the cyclic code $C$ is a $[184,183,2]_{23}$ MDS code. It follows from Proposition \ref{prop8} that $C^{\bot}\subseteq C$. From Theorem \ref{theo17}, we obtain a $[[184,182,2]]_{23}$ QEC MDS code.
\end{example}
\begin{remark}\label{rem3}
	Comparing with these QEC MDS codes in Examples \ref{exam9},\ref{exam10},\ref{exam11} and known QEC MDS codes in \cite{Grassl2004,Jin2010,Kai2013,Kai2014,Chen2015,Jin2017,Shi2017,Fang2020,Ball2021,Jin2022,Grassl2007}, it is easy to find that our QEC codes constructed in Examples \ref{exam9},\ref{exam10},\ref{exam11} are all new, meaning that they exhibit parameters distinct from those of previously known constructions.
\end{remark}
\section{EAQEC codes of length $2^rp^s$}\label{sec6}
 Entanglement-assisted quantum error-correcting (EAQEC) codes, as a special type of QEC codes, plays an important role in the field of quantum information. They combine the concepts of quantum error correction and quantum entanglement to improve the error correction ability and efficiency in quantum information processing. 
 
In recent years, the research of their construction has been widely concerned by many scholars.
Based on this, we will construct some EAQEC codes by using cyclic codes of length $2^rp^s$ over $\mathbb{F}_q$ in this section. In order to achieve this goal, we first introduce an important construction method of EAQEC codes, which has been given in \cite{Guenda2020}. 

\begin{theorem}\label{theo21}
	Let $C$ be an $[n,k,d]$ classical linear code. Denote $\textrm{Hull}(C)=C\cap C^{\bot}$ and $l=\textrm{dim}(\textrm{Hull}(C))$. Then, an $[[n,k-l,d;n-k-l]]_q$ EAQEC code is constructed.
\end{theorem}

Let $C$ be a repeated-root cyclic code of length $2^rp^s$ over $F_q$. Regarding the hull of $C$ and its dimension, we have the following two theorems. 

\begin{theorem}\label{theo22}
	Let $q$ be a prime power and has the form $q=2^ab+1$. Suppose that $C= \textbf{\Big\langle}\textbf{\big(} m_0(x)\textbf{\big)}^{j_0}\textbf{\big(}m_{2^{r-1}}(x)\textbf{\big)}^{j_{2^{r-1}}}\prod\limits_{i=2}^r\prod\limits_{s\in D_i}\textbf{\big(}m_{s2^{r-i}}(x)\textbf{\big)}^{j_{s2^{r-i}}}\textbf{\Big\rangle}$ is a cyclic code of length $2^rp^s$ over $\mathbb{F}_{q}$, where $0\leq j_h\leq p^s$ for each $h$. Denote $\textrm{Hull}(C)=C\cap C^{\bot}$ and $l=\textrm{dim}(\textrm{Hull}(C))$. Then, we have $$\textrm{Hull}(C)= \textbf{\Big\langle}\textbf{\big(} m_0(x)\textbf{\big)}^{\textrm{max}\{j_0,p^s-j_0\}}\textbf{\big(}m_{2^{r-1}}(x)\textbf{\big)}^{\textrm{max}\{j_{2^{r-1}},p^s-j_{2^{r-1}}\}}$$ $$\prod_{i=2}^r\prod_{s\in D_i}\textbf{\big(}m_{s2^{r-i}}(x)\textbf{\big)}^{\textrm{max}\{j_{s2^{r-i}},p^s-j_{-s2^{r-i}}\}}\textbf{\Big\rangle},$$  and  $$l=\begin{cases} 
	2^rp^s-(\textrm{max}\{j_0,p^s-j_0\}+\textrm{max}\{j_{2^{r-1}},p^s-j_{2^{r-1}}\}+e_1),&\:if\:r\leq a;\\
	2^rp^s-(\textrm{max}\{j_0,p^s-j_0\}+\textrm{max}\{j_{2^{r-1}},p^s-j_{2^{r-1}}\}+e_2),&\:if\:r>a,
	\end{cases}$$ where $0\leq j_h\leq p^s$ for each $h$, $e_1=\sum\limits_{i=2}^r\sum\limits_{s\in D_i}\textrm{max}\{j_{s2^{r-i}},p^s-j_{-s2^{r-i}}\}$ and $e_2=\sum\limits_{i=2}^a\sum\limits_{s\in D_i}\textrm{max}\{j_{s2^{r-i}},p^s-j_{-s2^{r-i}}\}+\sum\limits_{i=a+1}^r\sum\limits_{s\in D_i}2^{i-a}\textrm{max}\{j_{s2^{r-i}},p^s-j_{-s2^{r-i}}\}$.
\end{theorem}
\begin{proof} 
According to Theorem 4.3.7 in \cite{Huffman2003}, for any two cyclic codes $C_1=\langle g_1(x)\rangle$ and  $C_2=\langle g_2(x)\rangle$ , we know that $C_1\cap C_2$ is still a cyclic code and has the generator polynomial $\textrm{lcm}(g_1(x),g_2(x))$. From Teorem \ref{theo4}, we get $C^{\bot}=  \textbf{\Big\langle}\textbf{\big(} m_0(x)\textbf{\big)}^{p^s-j_0}\textbf{\big(}m_{2^{r-1}}(x)\textbf{\big)}^{p^s-j_{2^{r-1}}}\prod\limits_{i=2}^r\prod\limits_{s\in D_i}\textbf{\big(}m_{s2^{r-i}}(x)\textbf{\big)}^{p^s-j_{-s2^{r-i}}}\textbf{\Big\rangle}$.	Therefore, we have  $$\textrm{Hull}(C)= \textbf{\Big\langle}\textbf{\big(} m_0(x)\textbf{\big)}^{\textrm{max}\{j_0,p^s-j_0\}}\textbf{\big(}m_{2^{r-1}}(x)\textbf{\big)}^{\textrm{max}\{j_{2^{r-1}},p^s-j_{2^{r-1}}\}}$$ $$\prod_{i=2}^r\prod_{s\in D_i}\textbf{\big(}m_{s2^{r-i}}(x)\textbf{\big)}^{\textrm{max}\{j_{s2^{r-i}},p^s-j_{-s2^{r-i}}\}}\textbf{\Big\rangle}.$$  
Furthermore, combined with the proof of Theorem \ref{theo9}, the dimension $l$ of $\textrm{Hull}(C)$ can be easily obtained. The proof is similar to that of Theorem \ref{theo9} and is omitted here
\end{proof} 

\begin{theorem}\label{theo23}
Let $q$ be a prime power and has the form $q=2^ab-1$. Suppose that $C= \textbf{\Big\langle}\textbf{\big(} m_0(x)\textbf{\big)}^{j_0}\textbf{\big(}m_{2^{r-1}}(x)\textbf{\big)}^{j_{2^{r-1}}}\textbf{\big(}m_{2^{r-2}}(x)\textbf{\big)}^{j_{2^{r-2}}}\prod\limits_{i=3}^r\prod\limits_{s\in O_i}\textbf{\big(}m_{s2^{r-i}}(x)\textbf{\big)}^{j_{s2^{r-i}}}\textbf{\Big\rangle}$ is a cyclic code of length $2^rp^s$ over $\mathbb{F}_{q}$, where $0\leq j_h\leq p^s$ for each $h$. Denote $\textrm{Hull}(C)=C\cap C^{\bot}$ and $l=\textrm{dim}(\textrm{Hull}(C))$. The following results hold.\\
	(1) When $a=2$, we have  $$\textrm{Hull}(C)= \textbf{\Big\langle}\textbf{\big(} m_0(x)\textbf{\big)}^{\textrm{max}\{j_0,p^s-j_0\}}\textbf{\big(}m_{2^{r-1}}(x)\textbf{\big)}^{\textrm{max}\{j_{2^{r-1}},p^s-j_{2^{r-1}}\}}$$ $$\textbf{\big(}m_{2^{r-2}}(x)\textbf{\big)}^{\textrm{max}\{j_{2^{r-2}},p^s-j_{2^{r-2}}\}}\prod_{i=3}^r\prod_{s\in O_i}\textbf{\big(}m_{s2^{r-i}}(x)\textbf{\big)}^{\textrm{max}\{j_{s2^{r-i}},p^s-j_{-s2^{r-i}}\}}\textbf{\Big\rangle},$$ 
	and  $$l=2^rp^s-(\textrm{max}\{j_0,p^s-j_0\}+\textrm{max}\{j_{2^{r-1}},p^s-j_{2^{r-1}}\}+2\textrm{max}\{j_{2^{r-2}},p^s-j_{2^{r-2}}\}+e_3),$$ where $0\leq j_h\leq p^s$ for each $h$ and $e_3=\sum\limits_{i=3}^r\sum\limits_{s\in D_i}2^{i-a}\textrm{max}\{j_{s2^{r-i}},p^s-j_{-s2^{r-i}}\}$. \\
	(2) When $a\geq3$, we have 
	$$\textrm{Hull}(C)= \textbf{\Big\langle}\textbf{\big(} m_0(x)\textbf{\big)}^{\textrm{max}\{j_0,p^s-j_0\}}\textbf{\big(}m_{2^{r-1}}(x)\textbf{\big)}^{\textrm{max}\{j_{2^{r-1}},p^s-j_{2^{r-1}}\}}$$ $$\textbf{\big(}m_{2^{r-2}}(x)\textbf{\big)}^{\textrm{max}\{j_{2^{r-2}},p^s-j_{2^{r-2}}\}}\prod_{i=3}^a\prod_{s\in O_i}\textbf{\big(}m_{s2^{r-i}}(x)\textbf{\big)}^{\textrm{max}\{j_{s2^{r-i}},p^s-j_{s2^{r-i}}\}}$$ 
	$$\prod_{i=a+1}^r\prod_{s\in O_i}\textbf{\big(}m_{s2^{r-i}}(x)\textbf{\big)}^{\textrm{max}\{j_{s2^{r-i}},p^s-j_{-s2^{r-i}}\}}\textbf{\Big\rangle},$$
		and  $$l=2^rp^s-(\textrm{max}\{j_0,p^s-j_0\}+\textrm{max}\{j_{2^{r-1}},p^s-j_{2^{r-1}}\}+2\textrm{max}\{j_{2^{r-2}},p^s-j_{2^{r-2}}\}+e_4),$$ where $0\leq j_h\leq p^s$ for each $h$ and $e_4=2\sum\limits_{i=3}^a\sum\limits_{s\in D_i}\textrm{max}\{j_{s2^{r-i}},p^s-j_{s2^{r-i}}\}+\sum\limits_{i=a+1}^r\sum\limits_{s\in D_i}2^{i-a}\textrm{max}\{j_{s2^{r-i}},p^s-j_{-s2^{r-i}}\}$.
\end{theorem}

Combining Theorem \ref{theo21} and Theorem \ref{theo22} , we can obtain the following EAQEC codes by using cyclic codes of length $2^rp^s$ over $F_q$. 

\begin{theorem}\label{theo24}
	Let $q$ be a prime power and has the form $q=2^ab+1$. Suppose that $C$ is a cyclic code of length $2^rp^s$ over $\mathbb{F}_{q}$ and can be represented in the form $C= \textbf{\Big\langle}\textbf{\big(} m_0(x)\textbf{\big)}^{j_0}\textbf{\big(}m_{2^{r-1}}(x)\textbf{\big)}^{j_{2^{r-1}}}\prod\limits_{i=2}^r\prod\limits_{s\in D_i}\textbf{\big(}m_{s2^{r-i}}(x)\textbf{\big)}^{j_{s2^{r-i}}}\textbf{\Big\rangle}$, where $0\leq j_h\leq p^s$ for each $h$. The following results hold. 
	
		(1) If $r\leq a$, then a $[[2^rp^s,l_1-l_2,d_H(C);l_1+l_2-2^rp^s]]_q$ EAQEC code is essential, where $l_1=\textrm{max}\{j_0,p^s-j_0\}+\textrm{max}\{j_{2^{r-1}},p^s-j_{2^{r-1}}\}+\sum\limits_{i=2}^r\sum\limits_{s\in D_i}\textrm{max}\{j_{s2^{r-i}},p^s-j_{-s2^{r-i}}\}$, $l_2=j_0+j_{2^{r-1}}+\sum\limits_{i=2}^{r}\sum\limits_{s\in D_i}j_{s2^{r-i}}$.
	
		(2) If $r>a$, then a $[[2^rp^s,l_3-l_4,d_H(C);l_3+l_4-2^rp^s]]_q$ EAQEC code is essential, where $l_3=\textrm{max}\{j_0,p^s-j_0\}+\textrm{max}\{j_{2^{r-1}},p^s-j_{2^{r-1}}\}+\sum\limits_{i=2}^a\sum\limits_{s\in D_i}\textrm{max}\{j_{s2^{r-i}},p^s-j_{-s2^{r-i}}\}+\sum\limits_{i=a+1}^r\sum\limits_{s\in D_i}2^{i-a}\textrm{max}\{j_{s2^{r-i}},p^s-j_{-s2^{r-i}}\}$ and $l_4=j_0+j_{2^{r-1}}+\sum\limits_{i=2}^{a}\sum\limits_{s\in D_i}j_{s2^{r-i}}+\sum\limits_{i=a+1}^{r}\sum\limits_{s\in D_i}(2^{i-a}j_{s2^{r-i}})$.
\end{theorem}

\begin{proof} 
(1) If $r\leq a$, according to the proof of Theorem \ref{theo9}, we know that the cyclic code $C$ has parameters $[2^rp^s,2^rp^s-j_0-j_{2^{r-1}}-\sum\limits_{i=2}^{r}\sum\limits_{s\in D_i}j_{s2^{r-i}},d_H(C)]_q$. From Theorem \ref{theo22}, we have $	\textrm{dim}(\textrm{Hull}(C))=2^rp^s-(\textrm{max}\{j_0,p^s-j_0\}+\textrm{max}\{j_{2^{r-1}},p^s-j_{2^{r-1}}\}+e_1)$. By Theorem \ref{theo21}, we obtain a $[[2^rp^s,(\textrm{max}\{j_0,p^s-j_0\}+\textrm{max}\{j_{2^{r-1}},p^s-j_{2^{r-1}}\}+e_1)-(j_0+j_{2^{r-1}}+\sum\limits_{i=2}^{r}\sum\limits_{s\in D_i}j_{s2^{r-i}}),d_H(C);(\textrm{max}\{j_0,p^s-j_0\}+\textrm{max}\{j_{2^{r-1}},p^s-j_{2^{r-1}}\}+e_1)+(j_0+j_{2^{r-1}}+\sum\limits_{i=2}^{r}\sum\limits_{s\in D_i}j_{s2^{r-i}})-2^rp^s]]_q$ EAQEC code.

(2) If $r> a$, it follows from the proof of Theorem \ref{theo9} that the cyclic code $C$ has parameters $[2^rp^s,2^rp^s-j_0-j_{2^{r-1}}-\sum\limits_{i=2}^{a}\sum\limits_{s\in D_i}j_{s2^{r-i}}-\sum\limits_{i=a+1}^{r}\sum\limits_{s\in D_i}(2^{i-a}j_{s2^{r-i}}),d_H(C)]_q$. From Theorem \ref{theo22}, we have $	\textrm{dim}(\textrm{Hull}(C))=2^rp^s-(\textrm{max}\{j_0,p^s-j_0\}+\textrm{max}\{j_{2^{r-1}},p^s-j_{2^{r-1}}\}+e_2)$. By Theorem \ref{theo21}, we obtain an EAQEC code with parameters $[[2^rp^s,(\textrm{max}\{j_0,p^s-j_0\}+\textrm{max}\{j_{2^{r-1}},p^s-j_{2^{r-1}}\}+e_2)-(j_0+j_{2^{r-1}}+\sum\limits_{i=2}^{a}\sum\limits_{s\in D_i}j_{s2^{r-i}}+\sum\limits_{i=a+1}^{r}\sum\limits_{s\in D_i}(2^{i-a}j_{s2^{r-i}})),d_H(C);(\textrm{max}\{j_0,p^s-j_0\}+\textrm{max}\{j_{2^{r-1}},p^s-j_{2^{r-1}}\}+e_2)+(j_0+j_{2^{r-1}}+\sum\limits_{i=2}^{a}\sum\limits_{s\in D_i}j_{s2^{r-i}}+\sum\limits_{i=a+1}^{r}\sum\limits_{s\in D_i}(2^{i-a}j_{s2^{r-i}}))-2^rp^s]]_q$.
\end{proof} 
Furthermore, combining Theorem \ref{theo21} and Theorem \ref{theo23} , one can obtain the following EAQEC codes by using cyclic codes of length $2^rp^s$ over $\mathbb{F}_q$. The proof process is identical to that of Theorem \ref{theo24}, so it is omitted.

\begin{theorem}\label{theo25}
	Let $q$ be a prime power and has the form $q=2^ab-1$. Suppose that $C$ is a cyclic code of length $2^rp^s$ over $\mathbb{F}_{q}$ and can be represented in the form $C= \textbf{\Big\langle}\textbf{\big(} m_0(x)\textbf{\big)}^{j_0}\textbf{\big(}m_{2^{r-1}}(x)\textbf{\big)}^{j_{2^{r-1}}}\textbf{\big(}m_{2^{r-2}}(x)\textbf{\big)}^{j_{2^{r-2}}}\prod\limits_{i=3}^r\prod\limits_{s\in O_i}\textbf{\big(}m_{s2^{r-i}}(x)\textbf{\big)}^{j_{s2^{r-i}}}\textbf{\Big\rangle}$, where $0\leq j_h\leq p^s$ for each $h$. The following results hold. 
	
	(1) When $a=2$, then a $[[2^rp^s,l_5-l_6,d_H(C);l_5+l_6-2^rp^s]]_q$ EAQEC code is essential, where $l_5=\textrm{max}\{j_0,p^s-j_0\}+\textrm{max}\{j_{2^{r-1}},p^s-j_{2^{r-1}}\}+2\textrm{max}\{j_{2^{r-2}},p^s-j_{2^{r-2}}\}+\sum\limits_{i=3}^r\sum\limits_{s\in D_i}2^{i-a}\textrm{max}\{j_{s2^{r-i}},p^s-j_{-s2^{r-i}}\}$ and  $l_6=j_0+j_{2^{r-1}}+2j_{2^{r-2}}+\sum\limits_{i=3}^r\sum\limits_{s\in O_i}2^{i-a}s2^{r-i}$.
	
	(2) If $a\geq3$, then a $[[2^rp^s,l_7-l_8,d_H(C);l_7+l_8-2^rp^s]]_q$ EAQEC code is essential, where $l_7=\textrm{max}\{j_0,p^s-j_0\}+\textrm{max}\{j_{2^{r-1}},p^s-j_{2^{r-1}}\}+2\textrm{max}\{j_{2^{r-2}},p^s-j_{2^{r-2}}\}+2\sum\limits_{i=3}^a\sum\limits_{s\in D_i}\textrm{max}\{j_{s2^{r-i}},p^s-j_{s2^{r-i}}\}+\sum\limits_{i=a+1}^r\sum\limits_{s\in D_i}2^{i-a}\textrm{max}\{j_{s2^{r-i}},p^s-j_{-s2^{r-i}}\}$ and $l_8=j_0+j_{2^{r-1}}+2j_{2^{r-2}}+2\sum\limits_{i=3}^a\sum\limits_{s\in D_i}j_{s2^{r-i}}+\sum\limits_{i=a+1}^r\sum\limits_{s\in D_i}2^{i-a}j_{s2^{r-i}}$.
\end{theorem}

\begin{example} \label{exam12}
 Let $a=2,b=3$, then $q=2^ab+1=13$ and $p=13$. Let $r=3,s=1$, then $n=2^rp^s=104$. By factorization, we have  $$x^{104}-1=(x+1)^{13}(x-1)^{13}(x+5)^{13}(x-5)^{13}(x^2+5)^{13}(x^2-5)^{13}.$$ Put $C=\textbf{\Big\langle}(x+1)^{7}(x-1)(x+5)^{2}(x^2+5)^{3}(x^2-5)\textbf{\Big\rangle}$. According to Theorem 7 in \cite{Pan2024}, one can get $d_H(C)=8$. From Theorem \ref{theo22}, we get $$\textrm{Hull}(C)=\textbf{\Big\langle}(x+1)^{7}(x-1)^{12}(x+5)^{13}(x-5)^{11}(x^2+5)^{12}(x^2-5)^{10} \textbf{\Big\rangle}$$ and $l=\textrm{dim}(\textrm{Hull}(C))=17$.  
 By Theorem \ref{theo24}, we obtain a $[[104,69,8;1]]_{13}$ EAQEC code.	
\end{example}

\begin{example} \label{exam13}
 Let $a=4,b=1$, then $q=2^ab+1=17$ and $p=17$. Let $r=3,s=1$, then $n=2^rp^s=136$. By factorization, we have   $$x^{136}-1=(x+1)^{17}(x-1)^{17}(x+2)^{17}(x-2)^{17}(x+4)^{17}(x-4)^{17}(x+8)^{17}(x-8)^{17}.$$  Put $C=\textbf{\Big\langle}(x+1)^{9}(x-1)(x+2)(x-2)^2(x+4)(x+8)(x-8)^3\textbf{\Big\rangle}$. According to Theorem 4 in \cite{Pan2024}, one can get $d_H(C)=8$. From Theorem \ref{theo22}, we get $$\textrm{Hull}(C)=\textbf{\Big\langle}(x+1)^{9}(x-1)^{16}(x+2)^{14}(x-2)^{16}(x+4)^{17}(x-4)^{16}(x+8)^{15}(x-8)^{16}\textbf{\Big\rangle}$$ and $l=\textrm{dim}(\textrm{Hull}(C))=17$.  
 By Theorem \ref{theo24}, we obtain a $[[136,101,8;1]]_{13}$ EAQEC code.	
\end{example}

\begin{example} \label{exam14}
	Let $a=2,b=3$, then $q=2^ab-1=11$ and $p=11$. Let $r=3,s=1$, then $n=2^rp^s=88$. By factorization, we have   $$x^{88}-1=(x+1)^{11}(x-1)^{11}(x^2+1)^{11}(x^2+3x+10)^{11}(x^2+8x+10)^{11}.$$  Put $C=\textbf{\Big\langle}(x+1)^4(x-1)^4(x^2+1)^5(x^2+3x+10)^{10}(x^2+8x+10)^2\textbf{\Big\rangle}$. According to Theorem 10 in \cite{Pan2024}, one can get $d_H(C)=11$. From Theorem \ref{theo22}, we get $$\textrm{Hull}(C)=\langle(x+1)^{7}(x-1)^{7}(x^2+1)^{6}(x^2+3x+10)^{10}(x^2+8x+10)^2\rangle$$ and $l=\textrm{dim}(\textrm{Hull}(C))=38$.  
	By Theorem \ref{theo24}, we obtain an $[[88,8,11;4]]_{13}$ EAQEC code.	
\end{example}

\begin{example} \label{exam15}
	Let $a=5,b=1$, then $q=2^ab-1=31$ and $p=31$. Let $r=3,s=1$, then $n=2^rp^s=248$. By factorization, we have   $$x^{248}-1=(x+1)^{31}(x-1)^{31}(x^2+1)^{31}(x^2+8x+1)^{31}(x^2+23x+1)^{31}.$$  Put $C=\textbf{\Big\langle}(x+1)^{16}(x^2+8x+1)\textbf{\Big\rangle}$. According to Theorem 10 in \cite{Pan2024}, one can get $d_H(C)=4$. From Theorem \ref{theo22}, we get $$\textrm{Hull}(C)=\textbf{\Big\langle}(x+1)^{16}(x-1)^{31}(x^2+1)^{31}(x^2+8x+1)^{31}(x^2+23x+1)^{30}\textbf{\Big\rangle}$$ and $l=\textrm{dim}(\textrm{Hull}(C))=17$.  
	By Theorem \ref{theo25}, we obtain an $[[248,213,4;1]]_{13}$ EAQEC code.	
\end{example}
\begin{remark}\label{rem4}
The Hamming distances of all cyclic codes in Examples \ref{exam12}, \ref{exam13}, \ref{exam14}, and \ref{exam15} can be determined using our latest research results in \cite{Pan2024}. Comparing the EAQEC codes constructed in Examples \ref{exam12}, \ref{exam13}, \ref{exam14}, and \ref{exam15} with the established EAQEC codes in \cite{Li2011, Fan2016, Chen2017, Chen2018, Lu2018, Guenda2018, Li2019, Fang2019, Wang2020, Guenda2020, Jin2021, Grassl2007}, it is evident that our EAQEC codes in Examples \ref{exam12}, \ref{exam13}, \ref{exam14}, and \ref{exam15} are all novel. This means that they possess distinct parameters compared to the previously known constructions.
\end{remark}

\section{Conclusion}\label{sec7}
To the best of our knowledge, Dinh et al. have constructed a large number of new QEC codes by using repeated-root cyclic codes of lengths $p^s,2p^s$ and $4p^s$ \cite{Dinhps,Dinh2ps,Dinh4ps}. Therefore, in this paper, we always assume $r\geq3$. In this paper, by studying the structure of all repeated-root cyclic codes of length $2^rp^s$ and their duals, we first give the necessary and sufficient conditions for them to be dual-containing codes. Secondly, based on CSS and Steane's constructions, we construct a series of QEC codes using these cyclic codes. Additionally, we provide a comprehensive list of all MDS cyclic codes of length $2^rp^s$ and utilize them to build numerous QEC MDS codes. Lastly, we create a series of EAQEC codes based on the hulls of these repeated-root cyclic codes. Furthermore, in order to illustrate our conclusions, we enumerate some specific examples. In Examples \ref{exam1}-\ref{exam8}, we construct a large amount of new QEC codes which means that they exhibit parameters distinct from those of previously known constructions. Some of these QEC codes are consistent with some known codes constructed in \cite{Grassl2007}. And some QEC codes are better than all QEC codes of the same length and Hamming distance given in \cite{Grassl2007}, that is, the dimension of our QEC codes is greater than the dimension of all the QEC codes with the same length and Hamming distance in \cite{Grassl2007}. In Examples \ref{exam9}-\ref{exam11}, we construct some new QEC MDS codes over $\mathbb{F}_{23},\mathbb{F}_{25}$ and $\mathbb{F}_{29}$. Finally, in Examples \ref{exam12}-\ref{exam15}, we give some new EAQEC codes to illustrate our results in Section \ref{sec7}.

\backmatter






\bmhead{Acknowledgments}
This work was supported in part by the National Natural Science Foundation of China under Grant 62201009, in part by the  Anhui Provincial Natural Science Foundation under Grant 2108085QA06, and in part by the the High-level Talent Research Foundation of Anhui Agricultural University under Grant rc382005.
\bmhead{Data availability}
This paper does not contain any other data.

\section*{Declarations}

\bmhead{Conflict of interest} The authors declare that they have no known competing financial interests or personal
relationships that could have appeared to influence the work reported in this paper.



\noindent







\end{document}